\documentclass[11pt]{article}

\usepackage[margin=1in]{geometry}

\usepackage{amsmath}
\usepackage{amssymb}
\usepackage{amsthm}

\usepackage{booktabs}
\usepackage{array}

\usepackage{listings}
\usepackage{xcolor}

\usepackage[colorlinks=true,linkcolor=blue,citecolor=blue,urlcolor=blue]{hyperref}

\usepackage{natbib}

\usepackage{enumitem}
\usepackage{microtype}
\usepackage{graphicx}

\newtheorem{theorem}{Theorem}[section]
\newtheorem{proposition}[theorem]{Proposition}
\newtheorem{corollary}[theorem]{Corollary}

\theoremstyle{definition}
\newtheorem{definition}[theorem]{Definition}
\theoremstyle{remark}
\newtheorem{remark}[theorem]{Remark}

\definecolor{codebg}{rgb}{0.95,0.95,0.95}
\definecolor{codegreen}{rgb}{0,0.6,0}
\definecolor{codepurple}{rgb}{0.58,0,0.82}
\definecolor{codeblue}{rgb}{0.0,0.0,0.7}

\lstdefinestyle{formal}{
  backgroundcolor=\color{codebg},
  basicstyle=\ttfamily\small,
  breaklines=true,
  frame=single,
  framerule=0pt,
  xleftmargin=1em,
  xrightmargin=1em,
  aboveskip=0.5em,
  belowskip=0.5em,
  mathescape=true,
  showstringspaces=false,
  literate={->}{$\to$}2 {=>}{$\Rightarrow$}2
}

\lstdefinestyle{elixir}{
  backgroundcolor=\color{codebg},
  basicstyle=\ttfamily\small,
  breaklines=true,
  frame=single,
  framerule=0pt,
  xleftmargin=1em,
  xrightmargin=1em,
  aboveskip=0.5em,
  belowskip=0.5em,
  keywordstyle=\color{codeblue}\bfseries,
  stringstyle=\color{codegreen},
  commentstyle=\color{gray}\itshape,
  morekeywords={defmodule,def,defp,do,end,case,with,if,else,fn,true,false,nil,when,defstruct},
  morecomment=[l]{\#},
  morestring=[b]",
  sensitive=true,
  showstringspaces=false,
}

\lstdefinestyle{pseudo}{
  backgroundcolor=\color{codebg},
  basicstyle=\ttfamily\small,
  breaklines=true,
  frame=single,
  framerule=0pt,
  xleftmargin=1em,
  xrightmargin=1em,
  aboveskip=0.5em,
  belowskip=0.5em,
  keywordstyle=\color{codeblue}\bfseries,
  commentstyle=\color{gray}\itshape,
  morekeywords={function,return,for,each,in,do,if,then,else,end,match,case,fail,with,let,and,or,not,reject,accept,raise,error},
  morecomment=[l]{//},
  morestring=[b]",
  sensitive=true,
  showstringspaces=false,
}

\lstdefinestyle{wasm}{
  backgroundcolor=\color{codebg},
  basicstyle=\ttfamily\small,
  breaklines=true,
  frame=single,
  framerule=0pt,
  xleftmargin=1em,
  xrightmargin=1em,
  aboveskip=0.5em,
  belowskip=0.5em,
  keywordstyle=\color{codeblue}\bfseries,
  commentstyle=\color{gray}\itshape,
  morekeywords={module,import,func,export,param,result,call,i32,i64,f32,f64,local,memory,type},
  morecomment=[l]{;;},
  sensitive=true,
  showstringspaces=false,
}

\title{Certified Purity for Cognitive Workflow Executors:\\
From Static Analysis to Cryptographic Attestation}

\author{Alan L. McCann\\
\textit{Mashin, Inc.}\\
\texttt{research@mashin.live}}

\date{April 2026}

\begin{document}

\maketitle

\begin{abstract}
We present a \emph{certified purity architecture} that converts governance enforcement in cognitive workflow systems from a runtime convention into a structural capability boundary. A prior three-layer governance architecture~\cite{mccann2026structural} proves governance completeness, provenance completeness, and the impossibility of ungoverned effects, conditional on the \emph{pure module constraint}: that step executors cannot perform effects. That constraint was enforced by module import graph analysis, which is insufficient against adversarial bypass on the BEAM virtual machine. This paper closes the gap through four mechanisms: (1)~a restricted WebAssembly compilation target where effect-producing instructions are structurally absent; (2)~\emph{purity certificates}, cryptographically signed proofs binding executor binaries to their import classifications; (3)~a \emph{runtime verification gate} that rejects uncertified executors before they enter the governance pipeline; and (4)~\emph{portable governance credentials} via remote attestation for cross-organizational verification. We prove four theorems: structural purity by construction, bypass elimination for all five BEAM bypass classes, certificate integrity, and gate completeness. The guarantee holds relative to an explicit Trusted Computing Base. Evaluation on four implemented executors shows verification latency of 39--42\,$\mu$s, full plan cycle under 400\,$\mu$s, runtime overhead under 0.4\% of a 100\,ms HTTP request, and zero determinism divergences across repeated invocations.
\end{abstract}


\section{Introduction}
\label{sec:introduction}

\subsection{The Governance Problem}

Cognitive workflow systems (platforms orchestrating AI model inference, tool use, and multi-step reasoning as programmable workflows) increasingly operate in high-stakes domains: finance, healthcare, law, and autonomous agent operation. Current systems (LangChain~\cite{langchain2022}, CrewAI~\cite{crewai2023}, AutoGPT~\cite{autogpt2023}) rely on convention-based governance where developers manually instrument logging, permission checks, and audit trails at each execution point. The guarantee depends on every developer following conventions perfectly.

A prior paper~\cite{mccann2026structural} introduced a three-layer governance architecture that eliminates convention dependence through structural design. At the execution layer, workflow step executors are restricted to pure functions that return declarative \emph{directives} (data structures expressing intended effects), and a single \emph{directive interpreter} processes every directive through a governance pipeline comprising trust verification, permission checks, and provenance recording. The architecture yields three safety properties as formal theorems:

\begin{itemize}[nosep]
  \item \textbf{Governance Completeness:} $\Box\; \forall d.\; (\mathit{effect\_performed}(d) \to \mathit{governance\_checked}(d))$
  \item \textbf{Provenance Completeness:} $\Box\; \forall d.\; (\mathit{effect\_performed}(d) \to \mathit{provenance\_recorded}(d))$
  \item \textbf{No Ungoverned Effects:} $\neg\Diamond\; \exists d.\; (\mathit{effect\_performed}(d) \wedge \neg\mathit{governance\_checked}(d))$
\end{itemize}

These theorems hold for all programs expressible in the system, including executors not yet written. (The temporal operators $\Box$ and $\Diamond$ follow Pnueli's temporal logic of programs~\cite{pnueli1977temporal}.)

\subsection{The Purity Gap}

Every safety theorem in the prior work depends on a single premise: the \emph{pure module constraint}, specifically that executor modules have no access to I/O modules and therefore cannot perform effects. The prior paper enforces this constraint through module import graph analysis and acknowledges the enforcement mechanism:

\begin{quote}
``The structural purity constraint. The executor module has no access to modules that perform I/O. It does not import, alias, or call any module that interfaces with LLM APIs, HTTP clients, filesystem operations, database connections, PubSub systems, or process supervision. This is enforced at the module dependency level, not by convention.''~\cite{mccann2026structural}
\end{quote}

Module-level enforcement is strictly stronger than convention; a developer cannot accidentally bypass it. But a determined adversary can intentionally bypass it. On the BEAM virtual machine~\cite{armstrong2003making} (Erlang/OTP), five classes of bypass exist:

\paragraph{Bypass Class 1: Dynamic dispatch.} \texttt{:erlang.apply(module, function, args)} invokes any function in any loaded module. No import is required. An executor that never imports \texttt{Req} or \texttt{HTTPoison} can still issue HTTP requests via \texttt{:erlang.apply(:httpc, :request, [url])}.

\paragraph{Bypass Class 2: Code evaluation.} \texttt{Code.eval\_string("File.write!(\char`\\"/tmp/exfil\char`\\", data)")} evaluates arbitrary Elixir code at runtime, with full access to the BEAM module environment.

\paragraph{Bypass Class 3: Native implemented functions (NIFs).} A NIF is a shared library loaded into the BEAM process space. NIFs execute native code with no BEAM-level restriction: they can open sockets, write files, and call arbitrary system APIs.

\paragraph{Bypass Class 4: Ports.} \texttt{:erlang.open\_port(\{:spawn, "curl http://evil.com"\}, [])} spawns an external process with full system access. The executor communicates with the process via message passing; no I/O module import is required.

\paragraph{Bypass Class 5: Dynamic module loading.} \texttt{:code.load\_binary(Module, filename, binary)} loads arbitrary BEAM bytecode at runtime, including modules with unrestricted I/O capabilities.

\medskip

Static analysis can detect known patterns of these bypasses. But by Rice's theorem~\cite{rice1953classes}, any non-trivial semantic property of programs is undecidable. No static analysis can determine, for all possible programs, whether they will perform effects via dynamic dispatch, code evaluation, or other indirect mechanisms. A sound overapproximation (conservatively rejecting any program that \emph{might} perform effects) is theoretically possible, but the BEAM's pervasive dynamic features (runtime atom construction, hot code loading, universal message passing) make such an overapproximation impractically coarse: it would reject the majority of legitimate executors that use standard Elixir idioms. The analysis detects violations; it cannot prevent them without unacceptable false positive rates.

This is the \emph{purity gap}: the formal model assumes purity as a property of executor modules, but the implementation enforces it through a mechanism (static analysis) that is necessary but insufficient against adversarial authors. The safety theorems are correct, but their foundational premise rests on the weakest link in the architecture.

\subsection{Contributions}

This paper makes the following contributions:

\begin{enumerate}
  \item \textbf{A formal threat model} (Section~\ref{sec:threat}) characterizing the five bypass classes, their mechanisms, and the fundamental limitation of static analysis on the BEAM.

  \item \textbf{A restricted compilation target} (Section~\ref{sec:architecture}) based on WebAssembly that eliminates all five bypass classes by construction: executors compiled to the restricted target can only invoke a whitelisted set of pure host functions, making effect bypass structurally impossible.

  \item \textbf{A purity certificate scheme} (Section~\ref{sec:architecture}) using Ed25519 signatures over the executor binary and its purity proof, with a formal binding property ensuring certificates cannot be transferred to modified artifacts.

  \item \textbf{A runtime verification gate} (Section~\ref{sec:architecture}) that integrates with the directive interpreter, rejecting any executor without a valid purity certificate before it enters the governance pipeline.

  \item \textbf{Formal proofs} (Section~\ref{sec:formal}) that:
  \begin{itemize}[nosep]
    \item The restricted target satisfies the pure module constraint by construction (Theorem~\ref{thm:structural-purity})
    \item Each of the five bypass classes is individually eliminated (Theorem~\ref{thm:bypass-elimination})
    \item The prior paper's safety theorems hold with the pure module constraint discharged mechanically rather than assumed (Corollary~\ref{cor:prior-theorems})
  \end{itemize}

  \item \textbf{Portable governance credentials} (Section~\ref{sec:attestation}) via remote attestation, extending purity certificates to a cross-organizational trust substrate for distributed AI execution, with a formal compatibility predicate for cross-org governance composition.

  \item \textbf{An explicit Trusted Computing Base definition} (Section~\ref{sec:tcb}) identifying the five runtime components that must be correct for the structural guarantee to hold, with hardening strategies that further reduce trust surface.

  \item \textbf{An implementation architecture} (Section~\ref{sec:implementation}) integrating BEAM, Wasmex, and the existing directive interpreter, with pseudocode for key algorithms and analysis of the build pipeline.
\end{enumerate}

\paragraph{Scope of ``certified.''}  The term \emph{certified purity} refers to the runtime certification gate's mechanical check of purity certificates, not to mechanized proof in the sense of Rocq or Lean.  The proofs in this paper are pen-and-paper; the \emph{certification} is the runtime's algorithmic confirmation that an executor's imports fall within the pure whitelist $\mathcal{W}$, performed at load time before every execution.  The companion paper~\cite{mccann2026structural} provides mechanized Rocq proofs for the underlying governance safety properties; mechanizing this paper's results (e.g., formalizing the WASM capability model and whitelist purity in a proof assistant) is identified as future work that would further strengthen confidence but is not required for the architectural guarantee.

\paragraph{Scope and companion papers.}
This paper addresses the executor-level enforcement gap: converting the pure module constraint from a convention to a structural guarantee. It does not address the formal governance properties themselves or their algebraic generalization, which are developed in companion papers.
\cite{mccann2026structural}~establishes the structural governance criterion and the two-boundary model whose pure module constraint this paper discharges.
\cite{mccann2026mechanized}~provides the mechanized Rocq proofs for the safety and invariance theorems that depend on this constraint.
\cite{mccann2026gcc}~proves that governed execution is semantically transparent: permitted programs compute the same results with or without governance, a property that holds precisely because executors are pure.
\cite{mccann2026algebraic}~lifts these results to a parametric algebraic framework with extraction to a verified NIF integrated into the BEAM runtime.
\cite{mccann2026provenance}~extends the governance lifecycle to the supply chain; attestation records from the purity certificates compose with its distribution provenance system.

\section{Background}
\label{sec:background}

\subsection{The Pure Execution Model}
\label{sec:bg-pure}

We summarize the pure execution model from the prior work~\cite{mccann2026structural}, using the same notation.

\begin{definition}[Executor~{\cite[Def.~1]{mccann2026structural}}]
\label{def:executor}
An executor is a module implementing a pure function:
\[
\mathit{plan} : \mathit{StepConfig} \times \mathit{Context} \to \mathit{Result} \times [\mathit{Directive}] \;\cup\; \mathit{Error}
\]
The function receives step configuration and execution context, and returns either a result with a list of directives, or an error.
\end{definition}

\begin{definition}[Directive~{\cite[Def.~2]{mccann2026structural}}]
\label{def:directive}
A directive is a member of a sum type $D$:
\begin{align*}
D = \;& \mathit{LLMCall}(\mathit{model}, \mathit{system}, \mathit{user}, \mathit{opts}) \\
  |\;& \mathit{HTTPRequest}(\mathit{method}, \mathit{url}, \mathit{headers}, \mathit{body}) \\
  |\;& \mathit{FileOp}(\mathit{operation}, \mathit{path}, \mathit{content}) \\
  |\;& \mathit{CallMachine}(\mathit{machine}, \mathit{inputs}) \\
  |\;& \mathit{MemoryOp}(\mathit{operation}, \mathit{key}, \mathit{value}, \mathit{scope}) \\
  |\;& \cdots
\end{align*}
Directives are data. They describe intent without performing it.
\end{definition}

\begin{definition}[Governance Pipeline~{\cite[Def.~3]{mccann2026structural}}]
\label{def:governance}
The governance pipeline is a function:
\[
\mathit{govern} : D \times G \to \mathit{Governed}(\mathit{Result}) \mid \mathit{Denied}(\mathit{Reason})
\]
where $G$ is a governance context. The pipeline is a composition:
\[
\mathit{govern}(d, g) = \mathit{trust}(d, g) \otimes \mathit{permission}(d, g) \otimes \mathit{phase}(d, g) \otimes \mathit{hooks}(d, g)
\]
where $\otimes$ is sequential composition with short-circuit on failure.
\end{definition}

\begin{definition}[Directive Interpreter~{\cite[Def.~4]{mccann2026structural}}]
\label{def:interpreter}
The interpreter is the sole function crossing the pure/effectful boundary:
\[
\mathit{interpret} : [D] \times G \to [\mathit{Result}]
\]
For each directive $d$:
\begin{align*}
\mathit{interpret}_1(d, g) = \;&\mathbf{let}\; \_\, = \mathit{trust}(d, g) \;\mathbf{in} & \text{-- verify trust ceiling} \\
  &\mathbf{let}\; \_\, = \mathit{permission}(d, g) \;\mathbf{in} & \text{-- check permissions} \\
  &\mathbf{let}\; \_\, = \mathit{phase}(d, g) \;\mathbf{in} & \text{-- validate execution phase} \\
  &\mathbf{let}\; \_\, = \mathit{pre\_hooks}(d, g) \;\mathbf{in} & \text{-- pre-execution hooks} \\
  &\mathbf{let}\; r = \mathit{execute}(d) \;\mathbf{in} & \text{-- perform the effect (\textsc{sole effect site})} \\
  &\mathbf{let}\; \_\, = \mathit{guardrails}(d, r, g) \;\mathbf{in} & \text{-- validate result} \\
  &\mathbf{let}\; \_\, = \mathit{record}(d, r, g) \;\mathbf{in} & \text{-- record provenance} \\
  &r
\end{align*}
\end{definition}

\paragraph{Key structural property.} No executor can bypass any stage of this pipeline because no executor performs effects. The interpreter is the only code that calls I/O modules. All safety theorems depend on this property.

\subsection{The Pure Module Constraint}
\label{sec:bg-constraint}

\begin{definition}[Pure Module Constraint~{\cite[Def.~6]{mccann2026structural}}]
\label{def:pure-constraint}
An executor module $E$ satisfies the pure module constraint if:
\[
\forall f \in \mathit{Type}(E) : f \in \mathbf{Pure}
\]
That is, every function accessible to $E$ is a pure function. No function in $E$'s transitive import closure performs I/O.
\end{definition}

The prior paper proves that under this constraint, any function definable in $E$ is pure (composition of pure functions is pure), and therefore $E$ cannot emit a directive and also execute it. The safety theorems follow from this constraint applied universally to all executors.

\subsection{Proof-Carrying Code}
\label{sec:bg-pcc}

Necula~\cite{necula1997pcc} introduced proof-carrying code (PCC), in which a code producer attaches a machine-checkable proof that the code satisfies a safety policy. The code consumer verifies the proof before execution, obtaining the safety guarantee without trusting the producer. PCC was originally applied to memory safety: proving that native code does not access memory outside its designated regions.

Our purity certificate adapts the PCC paradigm to a different property: \emph{effect purity} rather than memory safety. The proof artifact attests that the executor binary contains no import that could perform I/O, and the runtime verifies this attestation before loading the executor. The key difference from classical PCC is that our proof is structural rather than logical: it enumerates the module's imports and classifies each one against a known-pure whitelist, rather than proving a logical formula about memory access patterns. This structural approach is possible because WebAssembly's import mechanism is explicit and enumerable, with no implicit capabilities.

Appel~\cite{appel2001foundational} extended PCC to foundational proof-carrying code, where the proof is expressed in a minimal logic rather than trusting a complex verification condition generator. Our approach is analogous: the purity proof is simple enough (import enumeration and classification) that verification is trivially correct, requiring no complex inference engine.

Morrisett et al.~\cite{morrisett1999typed} demonstrated that type-theoretic safety properties can be preserved through compilation to typed assembly language (TAL). Our approach pursues a related but distinct strategy: rather than preserving a type-level purity property through compilation, we verify the \emph{absence} of effect capabilities in the compiled artifact directly. The purity guarantee is a post-compilation structural property of the bytecode, not a type-level invariant propagated from source.

\subsection{WebAssembly Capability Model}
\label{sec:bg-wasm}

WebAssembly (WASM)~\cite{haas2017wasm, rossberg2024wasm} is a portable bytecode format designed as a compilation target for high-level languages. Its security model is based on \emph{capability restriction}: a WASM module can only access capabilities explicitly provided to it through imports. Specifically:

\begin{enumerate}[nosep]
  \item A WASM module declares its imports explicitly in its binary format.
  \item The host environment provides implementations for these imports at instantiation time.
  \item The module \emph{cannot} call any function not provided through this mechanism.
  \item There is no equivalent to \texttt{:erlang.apply/3}, \texttt{Code.eval\_string/1}, or dynamic module loading.
  \item Memory is linear and bounded; there are no pointers to host memory or system calls.
\end{enumerate}

This capability model is not a convention; it is enforced by the WASM runtime's validation algorithm, which rejects modules that reference undeclared imports~\cite{haas2017wasm}. The restriction is a \emph{structural property} of the bytecode format, not a policy layered on top.

\paragraph{Capability security lineage.} WASM's import-based capability model stands in a tradition originating with Dennis and Van Horn's~\cite{dennis1966programming} foundational work on capability-based addressing, in which a process can only access resources for which it holds an explicit capability token. Miller~\cite{miller2006robust} refined this into the \emph{object-capability model}, where authority is conveyed exclusively through object references: a program can only affect the world through references it has been explicitly granted. WASM instantiates this principle at the bytecode level: a module's imports are its capabilities, the host controls which capabilities are granted, and no ambient authority exists. Our certified purity architecture exploits this property: by granting only pure capabilities (the whitelist $\mathcal{W}$), we obtain effect isolation as a direct consequence of the capability discipline.

\subsection{Bytecode Verification}
\label{sec:bg-bytecode}

The Java Virtual Machine~\cite{lindholm1999java} performs bytecode verification at class loading time, checking type safety, stack discipline, and access control. Java's verifier ensures that bytecode cannot forge references, access private fields from unauthorized classes, or violate the type system. However, Java bytecode verification does not address \emph{effect isolation}: a verified Java class can still perform arbitrary I/O through the standard library (file access, network operations, process execution) without restriction.

Our approach shares Java's strategy of load-time verification but targets a different property. Java's verifier checks type safety within an effectful language. Our verifier checks effect \emph{absence} within a restricted bytecode format. The combination of a capability-restricted bytecode format (WASM) with a purity-focused verification algorithm yields a guarantee that Java's verifier cannot provide: not merely type-safe execution, but provably effect-free execution.

Leroy's CompCert~\cite{leroy2009compcert} formally verified that compilation preserves semantic properties from source to machine code. Our strategy is complementary: rather than verifying that the compiler preserves purity, we verify that the \emph{compiled artifact lacks effect capabilities}. This post-compilation verification is compiler-agnostic: it holds regardless of how the WASM binary was produced.

\section{Threat Model and Bypass Analysis}
\label{sec:threat}

\subsection{Adversary Model}

We consider an \emph{adversarial executor author}: an entity that writes executor modules with the intent to bypass the governance pipeline and perform ungoverned effects. The adversary operates within the following constraints:

\begin{enumerate}[nosep]
  \item The adversary can write arbitrary code within an executor module.
  \item The adversary's module must pass static import graph analysis (otherwise it is rejected before execution).
  \item The adversary knows the architecture of the governance system.
  \item The adversary's goal is to perform an I/O effect (network request, file write, process spawn, etc.) that does not pass through the directive interpreter's governance pipeline.
\end{enumerate}

This is a stronger adversary than the typical developer who accidentally omits governance (the threat addressed by convention-based systems). Our adversary actively seeks to bypass governance while appearing compliant.

\subsection{Bypass Classification}

Table~\ref{tab:bypasses} classifies the five bypass mechanisms available on the BEAM virtual machine.

\begin{table}[ht]
\centering
\caption{BEAM bypass classes for the pure module constraint}
\label{tab:bypasses}
\small
\begin{tabular}{@{}p{0.5cm}p{2.8cm}p{4.5cm}p{3.8cm}@{}}
\toprule
\textbf{\#} & \textbf{Mechanism} & \textbf{Example} & \textbf{Static Detectability} \\
\midrule
1 & Dynamic dispatch & \texttt{:erlang.apply(m, f, a)} where \texttt{m} is computed at runtime & Detectable if \texttt{:erlang.apply} is called; undetectable if the atom is constructed indirectly \\
\addlinespace
2 & Code evaluation & \texttt{Code.eval\_string(s)} where \texttt{s} contains I/O calls & Detectable if \texttt{Code.eval\_string} is called; content of \texttt{s} is opaque to static analysis \\
\addlinespace
3 & NIFs & \texttt{:erlang.load\_nif(path, info)} loading a shared library with I/O & Detectable if NIF loading is called; NIF behavior is opaque \\
\addlinespace
4 & Ports & \texttt{:erlang.open\_port(\{:spawn, cmd\}, opts)} & Detectable if \texttt{open\_port} is called; port creation can be indirect \\
\addlinespace
5 & Dynamic module loading & \texttt{:code.load\_binary(m, f, b)} loading arbitrary bytecode & Detectable if \texttt{:code} functions are called; module content is opaque \\
\bottomrule
\end{tabular}
\end{table}

\subsection{Why Static Analysis Is Insufficient}

Static analysis is \emph{necessary but insufficient} for enforcing purity on the BEAM. We establish this through two observations.

\paragraph{Observation 1: Pattern detection.} A static analyzer can detect direct calls to known dangerous functions (\texttt{:erlang.apply/3}, \texttt{Code.eval\_string/1}, \texttt{:erlang.open\_port/2}, \texttt{:code.load\_binary/3}, \texttt{:erlang.load\_nif/2}). Such an analyzer rejects executors containing these call patterns. This catches the majority of bypass attempts.

\paragraph{Observation 2: Undecidability.} By Rice's theorem~\cite{rice1953classes}, for any non-trivial semantic property $P$ of programs, it is undecidable whether an arbitrary program satisfies $P$. ``The executor performs no I/O effects at runtime'' is a non-trivial semantic property (some programs satisfy it, others do not). Therefore no algorithm can decide this property for all possible executor programs. A sound overapproximation (rejecting all programs that \emph{might} violate purity) is theoretically constructible, but the BEAM's pervasive dynamism (runtime atom construction, universal message passing, hot code loading, and open module namespaces) means such an analysis would conservatively reject most idiomatic Elixir executors. The false positive rate renders the approach impractical as a governance mechanism.

\paragraph{The gap.} Static analysis catches \emph{violations}, programs that demonstrably break the pure module constraint. It cannot \emph{prevent} violations, that is, guarantee that no expressible program breaks the constraint. The distinction matters for adversarial authors: given any finite set of patterns the analyzer checks, a sufficiently creative adversary can construct a program that bypasses those patterns while still performing I/O.

Consider a concrete example. An executor module imports only standard library pure functions. At runtime, it computes the atom \texttt{:httpc} through string manipulation (\texttt{String.to\_atom("ht" <> "tc")}), then invokes \texttt{:erlang.apply(:httpc, :request, [url])}. The import graph shows no I/O modules. A pattern-matching analyzer that searches for \texttt{:erlang.apply} catches this variant. But the adversary can obscure the call further: storing the \texttt{:erlang} atom in a data structure, passing it through several function calls, and invoking the apply only in a specific conditional branch. The arms race between obfuscation and detection has no principled termination~\cite{rice1953classes}.

The certified purity architecture terminates this arms race. Instead of detecting bypass patterns with ever-more-sophisticated analysis, it eliminates the \emph{capability} to bypass. An executor compiled to the restricted target \emph{cannot} invoke \texttt{:erlang.apply/3} because \texttt{:erlang.apply/3} is not in the host function whitelist. The distinction is between \emph{detection} (finding programs that bypass) and \emph{prevention} (making bypass structurally impossible).

\section{Verified Purity Architecture}
\label{sec:architecture}

\subsection{Restricted Compilation Target}
\label{sec:restricted-target}

\begin{definition}[Executor IR]
\label{def:executor-ir}
The \emph{Executor IR} (Intermediate Representation) is a WebAssembly module that conforms to the following restriction: every import in the module's import section belongs to the \emph{host function whitelist} $\mathcal{W}$.
\end{definition}

\begin{definition}[Host Function Whitelist]
\label{def:whitelist}
The host function whitelist $\mathcal{W}$ is a finite, enumerated set of host functions partitioned into two classes:
\begin{align*}
\mathcal{W}_{\mathrm{data}} = \{&\;\mathit{mem\_alloc},\; \mathit{mem\_free},\; \mathit{mem\_copy},\; \\
  &\;\mathit{str\_concat},\; \mathit{str\_slice},\; \mathit{str\_len},\; \mathit{str\_encode\_utf8},\; \\
  &\;\mathit{int\_add},\; \mathit{int\_sub},\; \mathit{int\_mul},\; \mathit{int\_div},\; \\
  &\;\mathit{float\_add},\; \mathit{float\_sub},\; \mathit{float\_mul},\; \mathit{float\_div},\; \\
  &\;\mathit{list\_new},\; \mathit{list\_push},\; \mathit{list\_get},\; \mathit{list\_len},\; \\
  &\;\mathit{map\_new},\; \mathit{map\_put},\; \mathit{map\_get},\; \mathit{map\_keys},\; \\
  &\;\mathit{json\_encode},\; \mathit{json\_decode},\; \\
  &\;\mathit{ctx\_get},\; \mathit{ctx\_get\_step\_output},\; \mathit{ctx\_get\_input}\; \}
\end{align*}

\begin{align*}
\mathcal{W}_{\mathrm{dir}} = \{&\;\mathit{directive\_llm\_call},\; \mathit{directive\_llm\_call\_stream},\; \\
  &\;\mathit{directive\_http\_request},\; \mathit{directive\_file\_op},\; \\
  &\;\mathit{directive\_call\_machine},\; \mathit{directive\_memory\_op},\; \\
  &\;\mathit{directive\_db\_op},\; \mathit{directive\_exec\_op},\; \\
  &\;\mathit{directive\_emit\_event},\; \mathit{directive\_broadcast}\; \}
\end{align*}

\[
\mathcal{W} = \mathcal{W}_{\mathrm{data}} \cup \mathcal{W}_{\mathrm{dir}}
\]

Functions in $\mathcal{W}_{\mathrm{data}}$ perform pure data operations: memory management within the WASM linear memory, string manipulation, arithmetic, collection operations, JSON serialization, and context reads. Functions in $\mathcal{W}_{\mathrm{dir}}$ are directive constructors: they build directive data structures and append them to the executor's output directive list. Neither class performs I/O.
\end{definition}

\begin{remark}
The directive constructors in $\mathcal{W}_{\mathrm{dir}}$ \emph{do not execute} the directives. They construct data structures representing intended effects and append them to an output buffer. The host implementation of $\mathit{directive\_http\_request}$, for example, creates an $\mathit{HTTPRequest}$ directive struct and pushes it onto the directive list. It does not issue an HTTP request. The directive interpreter processes the directive list after the executor returns.
\end{remark}

\paragraph{WASI exclusion.} The WebAssembly System Interface (WASI)~\cite{wasi2024} extends WASM with host-provided capabilities for filesystem access, network sockets, clocks, and random number generation. WASI imports are \emph{not} members of $\mathcal{W}$. Any module importing WASI functions (e.g., \texttt{wasi\_snapshot\_preview1:fd\_write}) will fail the verification gate at classification validation (Definition~\ref{def:gate}, step~5), because those imports classify as $\mathit{disallowed}$. This is by design: WASI exists precisely to grant effect capabilities, and effect capabilities are exactly what the certified purity architecture excludes. The exclusion requires no special-case logic; it falls out of the whitelist discipline.

\begin{proposition}[Whitelist Purity]
\label{prop:whitelist-purity}
Every function in $\mathcal{W}$ is pure: it is total, deterministic (same inputs produce same outputs), and side-effect-free (it performs no I/O, modifies no global state outside WASM linear memory, and communicates no information to external systems).
\end{proposition}

\begin{proof}
We verify purity per host function category. Our notion of purity is \emph{I/O purity}: a function is pure if its return value is determined entirely by its arguments, it has no observable side effects outside the WASM linear memory instance, and it terminates for all valid inputs. WASM traps (e.g., from memory exhaustion in \texttt{mem\_alloc}) constitute deterministic abnormal termination, not non-termination; purity is preserved because the trap outcome is determined entirely by the inputs and memory state, produces no side effects, and is detectable by the host. Mutation of the executor's own linear memory is permitted because it is not observable outside the executor.

\emph{Memory operations} ($\mathcal{W}_{\mathrm{data}}$, memory category): These operate on the WASM module's linear memory, an isolated byte array allocated per-instance by the WASM specification (Section~5.3 of~\cite{haas2017wasm}). No cross-instance effects are possible. The WASM linear memory allocator maintains internal bookkeeping, but this state is (a)~confined to the executor's memory instance, (b)~not observable by any other component, and (c)~deterministic given the same allocation sequence.

\emph{String and arithmetic operations} ($\mathcal{W}_{\mathrm{data}}$, compute category): Pure functions over immutable values. No state, no I/O. These include encoding/decoding (UTF-8), numeric conversions, and mathematical operations that map inputs to outputs with no side channels.

\emph{JSON operations} ($\mathcal{W}_{\mathrm{data}}$, serialization category): Parsing returns a result type (success with parsed value, or error with description). Errors are values, not exceptions. Serialization is deterministic: the same data structure produces the same JSON string. No I/O is performed.

\emph{Collection operations} ($\mathcal{W}_{\mathrm{data}}$, data structure category): Pure transformations on immutable data structures (list construction, map lookup, filtering). No shared state between invocations; each operation returns a new value without modifying the input.

\emph{Context reads} ($\mathcal{W}_{\mathrm{data}}$, context category): Pure reads from an immutable data structure provided by the host at invocation start. The context is frozen before execution begins and cannot be modified by the executor.

\emph{Directive constructors} ($\mathcal{W}_{\mathrm{dir}}$): These allocate a data structure (the directive record) and append it to the executor's directive output buffer. The buffer is local to the current invocation and read only by the interpreter \emph{after} the executor returns. Directive constructors do not perform the described effect; they produce a \emph{description} of the effect. They modify only the output buffer, which is (a)~local to the invocation, (b)~not readable by the executor after writing, and (c)~consumed exactly once by the interpreter. Therefore they are side-effect-free in the relevant sense: they produce data, not effects.
\end{proof}

\begin{theorem}[Structural Purity]
\label{thm:structural-purity}
An executor compiled to the Executor IR satisfies the pure module constraint (Definition~\ref{def:pure-constraint}) by construction.
\end{theorem}

\begin{proof}
Let $E_w$ be an executor compiled to the Executor IR (a WASM module). By the WASM specification~\cite{haas2017wasm, rossberg2024wasm}, $E_w$ can only invoke functions declared in its import section. By Definition~\ref{def:executor-ir}, every import belongs to $\mathcal{W}$. By Proposition~\ref{prop:whitelist-purity}, every function in $\mathcal{W}$ is pure.

The WASM execution model provides no mechanism for $E_w$ to invoke functions outside its declared imports:
\begin{enumerate}[nosep]
  \item WASM has no dynamic dispatch instruction (no equivalent to \texttt{:erlang.apply/3}).
  \item WASM has no code evaluation instruction (no equivalent to \texttt{Code.eval\_string/1}).
  \item WASM modules cannot load native code (no equivalent to NIFs).
  \item WASM modules cannot spawn external processes (no equivalent to ports).
  \item WASM modules cannot load other modules at runtime (no equivalent to \texttt{:code.load\_binary/3}).
\end{enumerate}

The indirect call instruction (\texttt{call\_indirect}) invokes a function through a table. Table entries are populated by two mechanisms: (1)~the host, during module instantiation, may place imported functions into the table, and (2)~\texttt{elem} segments in the module may place references to functions defined within the module or declared in its import section. In both cases, every reachable function is either a host function from $\mathcal{W}$ (pure by Proposition~\ref{prop:whitelist-purity}) or a module-internal function. Module-internal functions can only invoke imports from $\mathcal{W}$ and other internal functions; by induction on the call depth, all reachable functions are pure. Therefore indirect calls cannot escape the purity boundary.

Therefore every function callable by $E_w$ is pure. By the same argument as the prior paper's Proposition~3~\cite{mccann2026structural}: composition of pure functions is pure, so any function definable in $E_w$ is pure. $E_w$ satisfies the pure module constraint.
\end{proof}

\paragraph{Significance.} The proof does not depend on analyzing the executor's code. It depends only on the structural properties of the WASM format and the purity of the whitelist. No matter what code the executor contains, the purity guarantee holds, because the capability to perform effects does not exist in the execution environment.

\subsection{Trusted Computing Base}
\label{sec:tcb}

The structural purity guarantee established by Theorem~\ref{thm:structural-purity} holds \emph{relative to} a set of trusted components. We make this trust boundary explicit.

\begin{definition}[Trusted Computing Base]
\label{def:tcb}
The Trusted Computing Base (TCB) for the certified purity architecture is the set of components that must operate correctly for the structural effect isolation guarantee to hold:
\begin{enumerate}[nosep]
  \item The \textbf{WASM runtime} (Wasmtime): must correctly isolate WASM modules and not provide undeclared capabilities.
  \item The \textbf{host function implementations}: each function exposed through $\mathcal{W}$ must be pure (total, deterministic, side-effect-free).
  \item The \textbf{whitelist definition} $\mathcal{W}$: must contain only pure functions. A function in $\mathcal{W}$ whose implementation performs effects would violate the premise of Proposition~\ref{prop:whitelist-purity}.
  \item The \textbf{verification gate}: must correctly validate signatures, artifact hashes, and whitelist currency, and must be the sole path to executor invocation.
  \item The \textbf{directive interpreter}: must correctly apply the governance pipeline before performing effects.
\end{enumerate}
\end{definition}

\begin{remark}[TCB as strength, not weakness]
Every secure architecture has a Trusted Computing Base. Hardware security architectures trust the CPU. Operating systems trust the kernel. Cryptographic systems trust their mathematical assumptions. The question is not whether a TCB exists, but how large it is and whether it is explicit.

The certified purity architecture \emph{minimizes} the TCB by moving trust from unbounded executor code to a small, reviewable runtime surface. Before certified purity, the trust surface includes every executor module ($\sim$37 modules, $\sim$17,000 lines) plus every future and third-party executor. After certified purity, the trust surface is the five TCB components above, each auditable, each signable, each subject to independent review. This is a reduction from unbounded to bounded trust.
\end{remark}

\paragraph{Formal scoping.} Theorem~\ref{thm:structural-purity} establishes: \emph{if} $\mathcal{W}$ contains only pure functions and the WASM runtime correctly enforces the capability model, \emph{then} every executor compiled to the Executor IR is pure. The ``if'' clause is the TCB assumption. The conclusion is structural. Stating the TCB explicitly transforms the guarantee from an implicit claim to a precise security model: structural effect isolation holds relative to an uncompromised runtime and correctly implemented host function layer.

\paragraph{TCB hardening.} The TCB can be further reduced through three strategies: (A)~\emph{structural constraints on host functions}, restricting host function implementations from importing I/O-capable modules so the directive interpreter remains the sole effect site; (B)~\emph{physical separation}, splitting the runtime into a pure execution host process (WASM + host functions, no I/O capability) and an effect interpreter process, enforcing purity at the OS process boundary; (C)~\emph{signed runtime builds}, including runtime binary hashes, whitelist hashes, and interpreter version hashes in attestation records, making runtime integrity part of the cryptographic chain. These strategies are future hardening, discussed in Section~\ref{sec:discussion}.

\subsection{Purity Certificate}
\label{sec:certificate}

A purity certificate is a cryptographic proof artifact that binds an executor binary to its purity classification.

\begin{definition}[Purity Proof]
\label{def:purity-proof}
A purity proof $\pi$ for an Executor IR module $E_w$ is a structured record:
\[
\pi = (\mathit{imports}(E_w),\; \mathit{classify}(\mathit{imports}(E_w)),\; \mathit{conclusion})
\]
where:
\begin{itemize}[nosep]
  \item $\mathit{imports}(E_w)$ is the complete list of imports declared in $E_w$'s binary, extracted by parsing the WASM import section.
  \item $\mathit{classify}(\mathit{imports}(E_w))$ maps each import to its classification: $\mathit{pure\_data}$ (in $\mathcal{W}_{\mathrm{data}}$), $\mathit{pure\_directive}$ (in $\mathcal{W}_{\mathrm{dir}}$), or $\mathit{disallowed}$ (not in $\mathcal{W}$).
  \item $\mathit{conclusion} \in \{\mathit{pure},\; \mathit{impure}\}$: $\mathit{pure}$ if and only if all imports classify as $\mathit{pure\_data}$ or $\mathit{pure\_directive}$.
\end{itemize}
\end{definition}

\begin{definition}[Purity Certificate]
\label{def:purity-cert}
A purity certificate $C$ for an Executor IR module $E_w$ is:
\[
C = (\mathit{artifact\_hash},\; \mathit{proof\_hash},\; \sigma,\; \mathit{metadata})
\]
where:
\begin{itemize}[nosep]
  \item $\mathit{artifact\_hash} = \mathrm{SHA256}(E_w)$, the hash of the executor binary.
  \item $\mathit{proof\_hash} = \mathrm{SHA256}(\pi)$, the hash of the purity proof.
  \item $\sigma = \mathrm{Ed25519Sign}(k_{\mathrm{priv}},\; \mathit{artifact\_hash} \,\|\, \mathit{proof\_hash})$, a signature over the concatenation of both hashes, using the certifier's Ed25519 private key $k_{\mathrm{priv}}$.
  \item $\mathit{metadata}$ includes the certifier identity, timestamp, whitelist version, and certificate format version.
\end{itemize}
\end{definition}

The Ed25519 signature scheme~\cite{bernstein2012ed25519} provides 128-bit security with 64-byte signatures and fast verification (tens of microseconds on modern hardware).

\begin{theorem}[Certificate Integrity]
\label{thm:cert-integrity}
A valid purity certificate for artifact $E_w$ cannot be used for a modified artifact $E'_w \neq E_w$.
\end{theorem}

\begin{proof}
Let $C = (\mathit{artifact\_hash}, \mathit{proof\_hash}, \sigma, \mathit{metadata})$ be a valid certificate for $E_w$.

The verifier computes $h' = \mathrm{SHA256}(E'_w)$. Since $E'_w \neq E_w$, by the collision resistance of SHA-256, $h' \neq \mathit{artifact\_hash}$ with overwhelming probability (probability of collision: $< 2^{-128}$).

The verifier checks that $\mathrm{Ed25519Verify}(k_{\mathrm{pub}},\; h' \,\|\, \mathit{proof\_hash},\; \sigma) = \mathit{true}$. Since $\sigma$ was computed over $\mathit{artifact\_hash} \,\|\, \mathit{proof\_hash}$ and $h' \neq \mathit{artifact\_hash}$, the verification message differs. By the existential unforgeability of Ed25519 under chosen-message attacks, verification fails with overwhelming probability.

Therefore $C$ is not valid for $E'_w$.
\end{proof}

\begin{corollary}[No Certificate Transfer]
\label{cor:no-transfer}
An adversary who obtains a valid purity certificate for a pure executor $E_w$ cannot use that certificate for a modified executor $E'_w$ that imports disallowed functions, without forging an Ed25519 signature.
\end{corollary}

\subsection{Runtime Verification Gate}
\label{sec:gate}

The runtime verification gate is the enforcement point where purity certificates are checked before an executor enters the governance pipeline.

\begin{definition}[Verification Gate]
\label{def:gate}
The verification gate is a function:
\[
\mathit{gate} : E_w \times C \times \pi \to \mathit{Accept}(E_w) \mid \mathit{Reject}(\mathit{Reason})
\]
that performs the following checks in sequence (short-circuiting on failure):
\begin{enumerate}[nosep]
  \item \textbf{Signature verification:} Verify $\mathrm{Ed25519Verify}(k_{\mathrm{pub}},\; \mathrm{SHA256}(E_w) \,\|\, \mathrm{SHA256}(\pi),\; \sigma)$.
  \item \textbf{Artifact binding:} Verify $\mathrm{SHA256}(E_w) = C.\mathit{artifact\_hash}$.
  \item \textbf{Proof binding:} Verify $\mathrm{SHA256}(\pi) = C.\mathit{proof\_hash}$.
  \item \textbf{Proof validation:} Parse $E_w$'s import section independently and verify it matches $\pi.\mathit{imports}$.
  \item \textbf{Classification validation:} For each import in $\pi.\mathit{imports}$, verify the classification against the current whitelist $\mathcal{W}$.
  \item \textbf{Conclusion validation:} Verify $\pi.\mathit{conclusion} = \mathit{pure}$.
\end{enumerate}
\end{definition}

\paragraph{Defense in depth.} Steps 1--3 verify the certificate's cryptographic integrity. Steps 4--6 re-derive the purity proof from the artifact itself, independently of the certificate. An adversary would need to both forge a signature \emph{and} somehow alter the WASM module's import section without changing its hash, both of which are cryptographically infeasible. Step 4 is the critical defense-in-depth measure: even if the certificate were somehow forged, the gate independently verifies that the module's actual imports match the claimed proof.

\begin{theorem}[Gate Completeness]
\label{thm:gate-completeness}
Every executor that enters the governance pipeline has a valid purity certificate.
\end{theorem}

\begin{proof}
The directive interpreter is modified to invoke $\mathit{gate}(E_w, C, \pi)$ before executing any WASM-based executor's $\mathit{plan}$ function. If $\mathit{gate}$ returns $\mathit{Reject}$, the interpreter does not invoke $\mathit{plan}$ and returns an error to the caller. If $\mathit{gate}$ returns $\mathit{Accept}$, the interpreter proceeds with invocation.

This is a sequential precondition: $\mathit{plan}$ execution is gated behind $\mathit{gate}$. There is no code path from WASM executor instantiation to $\mathit{plan}$ invocation that does not pass through $\mathit{gate}$. Therefore every WASM-based executor that enters the governance pipeline (by having its $\mathit{plan}$ function invoked) has passed the verification gate.
\end{proof}

The pseudocode for the runtime gate is given as Algorithm~\ref{alg:gate}.

\begin{figure}[ht]
\begin{lstlisting}[style=pseudo,caption={Runtime Verification Gate},label=alg:gate]
function verify_and_load(wasm_binary, certificate, proof)
  // Step 1: Signature verification
  let message = SHA256(wasm_binary) || SHA256(proof)
  if not Ed25519Verify(certifier_pubkey, message, certificate.signature) then
    reject("invalid signature")
  end

  // Step 2: Artifact binding
  if SHA256(wasm_binary) != certificate.artifact_hash then
    reject("artifact hash mismatch")
  end

  // Step 3: Proof binding
  if SHA256(proof) != certificate.proof_hash then
    reject("proof hash mismatch")
  end

  // Step 4: Independent import extraction
  let actual_imports = parse_wasm_imports(wasm_binary)
  if actual_imports != proof.imports then
    reject("import mismatch between proof and binary")
  end

  // Step 5: Classification validation
  for each import in actual_imports do
    if import not in WHITELIST then
      reject("disallowed import: " || import)
    end
  end

  // Step 6: Conclusion validation
  if proof.conclusion != "pure" then
    reject("proof conclusion is not pure")
  end

  // All checks passed: instantiate WASM module
  let instance = wasm_instantiate(wasm_binary, host_functions)
  return accept(instance)
end
\end{lstlisting}
\end{figure}

\subsection{Provenance Integration}
\label{sec:provenance-integration}

The prior paper~\cite{mccann2026structural} defines the execution hash chain for provenance:
\begin{align*}
\mathit{execution\_hash}(s_i) = \mathrm{SHA256}(&\;\mathit{directive\_hash}(s_i) \,\|\, \mathit{governance\_hash}(s_i) \\
  &\;\|\, \mathit{result\_hash}(s_i) \,\|\, \mathit{execution\_hash}(s_{i-1}))
\end{align*}

We extend this to include the purity certificate:
\begin{align*}
\mathit{execution\_hash}_{\mathrm{vp}}(s_i) = \mathrm{SHA256}(&\;\mathit{directive\_hash}(s_i) \,\|\, \mathit{governance\_hash}(s_i) \\
  &\;\|\, \mathit{result\_hash}(s_i) \,\|\, \mathit{purity\_cert\_hash}(s_i) \\
  &\;\|\, \mathit{execution\_hash}_{\mathrm{vp}}(s_{i-1}))
\end{align*}

where $\mathit{purity\_cert\_hash}(s_i) = \mathrm{SHA256}(C_i)$ for the certificate of the executor that produced step $s_i$'s directives.

\paragraph{Auditor capabilities.} An auditor examining a provenance record can now verify not only that every effect was governed (governance hash) and that outputs are consistent (result hash), but also that the executor that produced the directives was provably pure at the time of execution (purity certificate hash). The auditor can retrieve the certificate, verify its signature, and independently validate the purity proof, all without access to the executor's source code.

The extended run provenance record becomes:
\[
\mathit{run\_hash}_{\mathrm{vp}} = \mathrm{SHA256}(\mathit{machine\_version\_hash} \,\|\, \mathit{input\_hash} \,\|\, \mathit{execution\_hash}_{\mathrm{vp}}(s_n) \,\|\, \mathit{output\_hash})
\]

\section{Formal Properties}
\label{sec:formal}

We now establish the formal properties of the certified purity architecture. The key result is that the prior paper's safety theorems hold with a strictly stronger foundation.

\subsection{Structural Purity (Restated)}

\begin{theorem}[Structural Purity (Restated)]
\label{thm:structural-purity-restated}
Under the certified purity architecture, executors compiled to the Executor IR satisfy the pure module constraint (Definition~\ref{def:pure-constraint}) by construction, not by static analysis.
\end{theorem}

\begin{proof}
This is Theorem~\ref{thm:structural-purity} restated for emphasis. The proof establishes that the pure module constraint is a structural consequence of the WASM capability model and the host function whitelist, independent of the executor's code content.

Formally, let $\mathcal{E}_w$ be the set of all syntactically valid Executor IR modules. For all $E_w \in \mathcal{E}_w$:
\[
\forall f \in \mathit{Callable}(E_w) : f \in \mathcal{W}
\]
where $\mathit{Callable}(E_w)$ is the set of all functions invocable by $E_w$ during execution (imports, indirect calls via table, internal functions). Since $\mathcal{W} \subset \mathbf{Pure}$ (Proposition~\ref{prop:whitelist-purity}), and internal WASM functions are pure computations over linear memory:
\[
\forall f \in \mathit{Callable}(E_w) : f \in \mathbf{Pure}
\]
which is exactly the pure module constraint.
\end{proof}

\subsection{Certificate Integrity (Restated)}

\begin{theorem}[Certificate Integrity (Restated)]
\label{thm:cert-integrity-restated}
A valid purity certificate for artifact $E_w$ cannot be transferred to a modified artifact $E'_w \neq E_w$, assuming the collision resistance of SHA-256 and the existential unforgeability of Ed25519.
\end{theorem}

\begin{proof}
See Theorem~\ref{thm:cert-integrity}. The security reduces to standard cryptographic assumptions.
\end{proof}

\subsection{Gate Completeness (Restated)}

\begin{theorem}[Gate Completeness (Restated)]
\label{thm:gate-completeness-restated}
Every executor that enters the governance pipeline has a valid purity certificate, and therefore satisfies the pure module constraint.
\end{theorem}

\begin{proof}
By Theorem~\ref{thm:gate-completeness}, every WASM-based executor passes the verification gate before invocation. By Definition~\ref{def:gate}, passing the gate requires a valid certificate. By Definition~\ref{def:purity-cert} and Theorem~\ref{thm:cert-integrity}, a valid certificate binds to an artifact that satisfies the pure module constraint. By Theorem~\ref{thm:structural-purity}, the artifact satisfies the constraint by construction.
\end{proof}

\subsection{Discharge of the Pure Module Constraint}

For self-containment, we restate the safety theorems from the prior paper~\cite{mccann2026structural} that depend on the pure module constraint:

\begin{enumerate}[nosep]
  \item \textbf{Governance Completeness:} $\Box\; \forall d.\; (\mathit{effect\_performed}(d) \to \mathit{governance\_checked}(d))$. Every effect that occurs in the system has passed through the governance pipeline.
  \item \textbf{Provenance Completeness:} $\Box\; \forall d.\; (\mathit{effect\_performed}(d) \to \mathit{provenance\_recorded}(d))$. Every effect that occurs has a corresponding entry in the provenance chain.
  \item \textbf{No Ungoverned Effects:} $\neg\Diamond\; \exists d.\; (\mathit{effect\_performed}(d) \wedge \neg\mathit{governance\_checked}(d))$. There is no execution path that produces effects without governance.
\end{enumerate}

\noindent All three are proved in~\cite{mccann2026structural} under the assumption that executors satisfy the pure module constraint (Definition~\ref{def:pure-constraint}). The core safety result ($\texttt{gov\_safe}$) is mechanized in Rocq with zero admitted lemmas~\cite{mccann2026mechanized}.\footnote{The Rocq development is available at \url{https://github.com/mashin-live/governance-proofs}.} The present paper discharges the assumption by construction.

\begin{corollary}[Prior Theorems Hold for Tier~1 Executors]
\label{cor:prior-theorems}
For executions where all step executors are Tier~1 (WASM with valid purity certificates), the safety theorems above hold with the pure module constraint discharged by construction rather than assumed. For mixed-tier deployments (Section~\ref{sec:trust-tiers}), the prior paper's guarantees for Tier~2/3 executors continue to rely on the original static analysis or convention-based enforcement.
\end{corollary}

\begin{proof}
The prior paper's safety proofs have the following structure:

\begin{enumerate}[nosep]
  \item \textbf{Premise:} All executors satisfy the pure module constraint (Definition~\ref{def:pure-constraint}).
  \item \textbf{Lemma:} Under the constraint, executors cannot perform effects (Proposition~2 in prior paper).
  \item \textbf{Lemma:} The interpreter is the sole effect site, and its structure ensures governance checking and provenance recording for every effect (Definition~\ref{def:interpreter}).
  \item \textbf{Conclusion:} Governance Completeness, Provenance Completeness, and No Ungoverned Effects hold for all executors.
\end{enumerate}

The certified purity architecture replaces Premise~1 with a construction: Theorem~\ref{thm:gate-completeness-restated} establishes that every executor entering the pipeline satisfies the constraint. The remainder of the proof is unchanged. The theorem statements are identical; only the justification for the premise is strengthened from ``enforced by static analysis'' to ``enforced by construction through capability restriction and cryptographic attestation.''
\end{proof}

\paragraph{What changed and what did not.} The safety theorems are the same three properties stated in the same temporal logic. The formal model is unchanged. The definitions of executor, directive, governance pipeline, and interpreter are identical. What changed is the \emph{confidence} in Premise~1: it was previously an assumption about the implementation (``module import graph analysis prevents bypass''), and is now a theorem about the architecture (``the WASM capability model and cryptographic attestation make bypass structurally impossible'').

\subsection{Bypass Elimination}

\begin{theorem}[Bypass Elimination]
\label{thm:bypass-elimination}
For each of the five BEAM bypass classes (Table~\ref{tab:bypasses}), the restricted compilation target makes the bypass structurally impossible.
\end{theorem}

\begin{proof}
We prove each case individually.

\paragraph{Bypass Class 1: Dynamic dispatch (\texttt{:erlang.apply/3}).}
The WASM instruction set contains no dynamic dispatch mechanism that takes a module name and function name as runtime values and invokes the corresponding function. WASM's \texttt{call} instruction references a function by a static index into the module's function table. WASM's \texttt{call\_indirect} instruction invokes a function through a runtime-determined table index, but the table is populated by the host at instantiation time with functions from $\mathcal{W}$ only. There is no instruction that accepts a string or atom identifying a function and resolves it to a callable target. Therefore dynamic dispatch to arbitrary functions is structurally impossible. \qed

\paragraph{Bypass Class 2: Code evaluation (\texttt{Code.eval\_string/1}).}
WASM modules execute pre-compiled bytecode. There is no instruction for evaluating source code at runtime. The module's code section is fixed at compilation time and validated at instantiation time. There is no mechanism for generating or loading new code during execution. Therefore runtime code evaluation is structurally impossible. \qed

\paragraph{Bypass Class 3: NIFs (native code loading).}
WASM modules execute within a sandboxed virtual machine. They have no access to the host's file system, shared library loading facilities, or native code execution mechanisms. The WASM specification provides no instruction for loading shared libraries or invoking native functions outside the declared import set. Therefore NIF-equivalent bypass is structurally impossible. \qed

\paragraph{Bypass Class 4: Ports (external process communication).}
WASM modules cannot spawn operating system processes. There is no instruction for process creation, pipe establishment, or inter-process communication outside the declared import set. The host function whitelist $\mathcal{W}$ contains no process-spawning function. Therefore port-equivalent bypass is structurally impossible. \qed

\paragraph{Bypass Class 5: Dynamic module loading (\texttt{:code.load\_binary/3}).}
WASM modules cannot load other WASM modules at runtime. Module instantiation is performed by the host, not by the module itself. There is no instruction for loading bytecode, instantiating modules, or linking additional imports during execution. Therefore dynamic module loading bypass is structurally impossible. \qed
\end{proof}

\subsection{Strength Comparison}

We summarize the guarantee strength at each enforcement level (Table~\ref{tab:strength}).

\begin{table}[ht]
\centering
\caption{Guarantee strength by enforcement mechanism}
\label{tab:strength}
\small
\begin{tabular}{@{}lcccc@{}}
\toprule
\textbf{Property} & \textbf{Convention} & \textbf{Static Analysis} & \textbf{Verified Purity} \\
\midrule
Detection of known patterns & \textsf{no} & \textsf{yes} & \textsf{yes} \\
Detection of unknown patterns & \textsf{no} & \textsf{no} & \textsf{n/a (prevented)} \\
Prevention of adversarial bypass & \textsf{no} & \textsf{no} & \textsf{yes} \\
Cryptographic attestation & \textsf{no} & \textsf{no} & \textsf{yes} \\
Auditor-verifiable & \textsf{no} & \textsf{partial} & \textsf{yes} \\
Holds for all future executors & \textsf{no} & \textsf{partially} & \textsf{yes} \\
\bottomrule
\end{tabular}
\end{table}

\subsection{Whitelist Evolution}
\label{sec:whitelist-evolution}

We establish that whitelist evolution preserves purity guarantees monotonically, ensuring that certificates remain valid as the whitelist grows and are correctly invalidated when it shrinks.

\begin{proposition}[Monotonic Purity Preservation]
\label{prop:monotonic-purity}
Let $\mathcal{W}_{v_1}$ and $\mathcal{W}_{v_2}$ be whitelist versions where $v_2 > v_1$. If $\mathcal{W}_{v_1} \subseteq \mathcal{W}_{v_2}$ (every function in $v_1$ is also present and pure in $v_2$), then any executor with a valid purity certificate against $\mathcal{W}_{v_1}$ remains pure under $\mathcal{W}_{v_2}$.

Conversely, if a function $f$ is present in $\mathcal{W}_{v_1}$ but removed from $\mathcal{W}_{v_2}$ (because $f$ was determined to be impure or unnecessary), then certificates generated against $\mathcal{W}_{v_1}$ that include $f$ in their import proofs are invalidated by the verification gate's whitelist currency check.
\end{proposition}

\begin{proof}
For the forward direction: if $\mathcal{W}_{v_1} \subseteq \mathcal{W}_{v_2}$, then every import classified as pure under $v_1$ remains classifiable as pure under $v_2$. The purity proof remains valid because no previously-pure import has become impure or been removed.

For the converse: the verification gate's classification validation (Definition~\ref{def:gate}, step 5) re-checks each import against the \emph{current} whitelist. If $f \notin \mathcal{W}_{v_2}$, classification fails for any certificate whose proof includes $f$. The executor is rejected. Additionally, the whitelist currency check rejects certificates whose whitelist version falls outside the acceptable range, providing a second defense against stale certificates.
\end{proof}

\section{Portable Governance Credentials}
\label{sec:attestation}

Certified purity is a local guarantee. Attestation makes it portable. Together, they define a cross-organizational trust substrate for governed AI execution.

The purity certificate proves to a local runtime that an executor is pure. But in distributed execution environments, where machines call other machines across organizational boundaries, a local proof is insufficient. Organization A, executing a machine that invokes Organization B's machine, needs cryptographic evidence (not trust) that B's executors are pure and that governance was applied. Remote attestation provides this evidence, converting purity certificates into \emph{portable governance credentials} that are independently verifiable by any party holding the relevant public keys.

\subsection{Attestation Record}

\begin{definition}[Attestation Record]
\label{def:attestation}
A remote attestation record $A$ for an executor $E_w$ is a structured document:
\[
A = (C,\; \pi,\; \mathit{env},\; \sigma_A)
\]
where:
\begin{itemize}[nosep]
  \item $C$ is the purity certificate (Definition~\ref{def:purity-cert}).
  \item $\pi$ is the purity proof (Definition~\ref{def:purity-proof}).
  \item $\mathit{env}$ is the execution environment descriptor: WASM runtime identity and version, host function whitelist version, and the public keys of accepted certifiers.
  \item $\sigma_A = \mathrm{Ed25519Sign}(k_{\mathrm{env}},\; \mathrm{SHA256}(C) \,\|\, \mathrm{SHA256}(\pi) \,\|\, \mathrm{SHA256}(\mathit{env}))$, a signature by the executing environment's key $k_{\mathrm{env}}$.
\end{itemize}
\end{definition}

The attestation record carries two signatures: the certifier's signature (in $C$) attests that the artifact is pure, and the environment's signature (in $\sigma_A$) attests that the environment verified the certificate and executed the artifact under the specified whitelist.

\subsection{Verification Protocol}

A remote verifier (e.g., an organization receiving execution results from a partner's system) performs the following steps:

\begin{enumerate}[nosep]
  \item \textbf{Trust establishment:} Verify that $k_{\mathrm{env}}$ (the environment's public key) is in the verifier's trust store.
  \item \textbf{Environment signature:} Verify $\mathrm{Ed25519Verify}(k_{\mathrm{env}},\; \mathrm{SHA256}(C) \,\|\, \mathrm{SHA256}(\pi) \,\|\, \mathrm{SHA256}(\mathit{env}),\; \sigma_A)$.
  \item \textbf{Certificate verification:} Verify the purity certificate $C$ as in Definition~\ref{def:gate}, steps 1--6.
  \item \textbf{Environment policy:} Verify that $\mathit{env}$ meets the verifier's requirements (e.g., minimum whitelist version, approved WASM runtime).
\end{enumerate}

If all steps pass, the verifier has cryptographic assurance that the executor was provably pure and was executed in a conforming environment, without trusting the executor author, the certifier, or the execution environment individually.

\subsection{Portable Trust Credentials}

An attestation record is a \emph{portable trust credential}: it can be transmitted alongside execution results, stored in provenance records, and verified by any party holding the relevant public keys. This enables:

\begin{itemize}[nosep]
  \item \textbf{Cross-organizational execution:} Organization A can execute a machine containing executors authored by Organization B, verify their purity certificates, and provide cryptographic proof to its own auditors that governance was complete.
  \item \textbf{Regulatory compliance:} An attestation record provides a machine-verifiable proof of purity that a regulatory auditor can check without access to source code or build infrastructure.
  \item \textbf{Supply chain integration:} Attestation records compose with the distribution provenance system~\cite{mccann2026provenance} (forthcoming), extending the four-phase lifecycle chain (definition, distribution, capability, execution) with a cryptographic purity proof at the capability-execution boundary.
\end{itemize}

\subsection{Cross-Organizational Governance Composition}

When Organization A's machine calls Organization B's machine (via the $\mathit{CallMachine}$ directive), the governance pipeline must verify purity across the organizational boundary.

The calling machine's provenance includes the called machine's attestation record. The provenance chain is:
\begin{align*}
\mathit{cross\_org\_hash} = \mathrm{SHA256}(&\;\mathit{caller\_run\_hash} \,\|\, \mathit{callee\_attestation\_hash} \\
  &\;\|\, \mathit{callee\_run\_hash})
\end{align*}

An auditor examining the caller's provenance can verify the entire cross-organizational chain: the caller's purity, the callee's purity (via attestation), and the governance of both execution paths.

\paragraph{Open question: governance completeness under composition.} The purity guarantee established in this paper is \emph{per-executor}: each executor in isolation satisfies the pure module constraint. When Organization A's machine calls Organization B's machine, both sides independently verify purity for their own executors. However, we do not prove that governance completeness \emph{composes} across organizational boundaries, that is, that the conjunction of A's governance completeness and B's governance completeness implies governance completeness for the composed execution. The gap lies in the cross-boundary directive handoff: A's interpreter produces a $\mathit{CallMachine}$ directive, but the governance of B's execution is attested after the fact, not structurally guaranteed before invocation. A formal composition theorem would require modeling the temporal relationship between A's governance check and B's execution, and is identified as future work.

\subsection{Compatibility Predicate}

Cross-organizational composition requires a formal compatibility condition. Without it, composition is conditional and ad hoc. With it, organizations define explicit, machine-evaluable trust policies.

\begin{definition}[Cross-Organizational Compatibility]
\label{def:compatibility}
Organization $A$ accepts an attestation record from Organization $B$ if and only if:
\begin{align*}
\mathrm{Compatible}(A, B) \iff \;&B.\mathit{whitelist\_hash} \in A.\mathit{accepted\_whitelists} \\
  \wedge\; &B.\mathit{runtime\_identity} \in A.\mathit{trusted\_runtimes} \\
  \wedge\; &B.\mathit{certifier} \in A.\mathit{trusted\_certifiers} \\
  \wedge\; &B.\mathit{whitelist\_version} \geq A.\mathit{minimum\_required}
\end{align*}
\end{definition}

The four conjuncts ensure: (1)~B's whitelist is recognized by A; (2)~B's runtime is in A's trust store; (3)~B's certifier is accepted by A; (4)~B's whitelist version meets A's minimum freshness requirement. Each organization independently configures its own compatibility policy. The compatibility predicate is evaluated mechanically; it is not a bilateral trust agreement but a unilateral, machine-checkable acceptance criterion.

\paragraph{Trust substrate.} The combination of purity certificates (local governance proof), attestation records (portable governance credentials), and the compatibility predicate (cross-org acceptance criterion) constitutes a \emph{trust substrate for distributed AI execution}. Organizations compose governed machines across trust boundaries without source code review, bilateral agreements, or shared infrastructure. The substrate relies only on mathematics: the WASM capability model, the collision resistance of SHA-256, and the unforgeability of Ed25519.

\section{Implementation}
\label{sec:implementation}

\subsection{BEAM + Wasmex Integration}
\label{sec:impl-wasmex}

The implementation integrates WebAssembly execution into the existing BEAM-based runtime using Wasmex~\cite{wasmex2023}, an Elixir library that provides bindings to the Wasmtime~\cite{wasmtime2023} WASM runtime.

The architecture is as follows:

\begin{enumerate}[nosep]
  \item The directive interpreter identifies the executor type (BEAM module or WASM module) from the machine definition.
  \item For WASM executors, the interpreter invokes the verification gate (Algorithm~\ref{alg:gate}).
  \item Upon acceptance, the interpreter instantiates the WASM module via Wasmex, providing host function implementations from $\mathcal{W}$.
  \item The interpreter calls the executor's exported $\mathit{plan}$ function.
  \item The executor returns directives via the directive constructor host functions.
  \item The interpreter processes the returned directives through the governance pipeline, identically to BEAM-based executors.
\end{enumerate}

\paragraph{BEAM executors are unchanged.} Existing BEAM-based executors continue to operate under static analysis enforcement. The certified purity architecture is an \emph{additional} enforcement tier, not a replacement. Organizations can adopt it incrementally: starting with WASM executors for third-party or high-assurance contexts, while retaining BEAM executors for trusted internal development.

\subsection{Host Function Interface Design}
\label{sec:impl-host}

The host function interface bridges the WASM linear memory model with the BEAM's rich data types. The design follows a serialization protocol:

\begin{enumerate}[nosep]
  \item \textbf{Input serialization:} The interpreter serializes the $\mathit{StepConfig}$ and $\mathit{Context}$ into JSON, writes the JSON bytes into the WASM module's linear memory, and passes memory offsets and lengths to the $\mathit{plan}$ function.
  \item \textbf{Context access:} Host functions in $\mathcal{W}_{\mathrm{data}}$ (e.g., $\mathit{ctx\_get}$, $\mathit{ctx\_get\_step\_output}$) read from a host-side context structure and write results into linear memory.
  \item \textbf{Directive construction:} Host functions in $\mathcal{W}_{\mathrm{dir}}$ (e.g., $\mathit{directive\_http\_request}$) read directive parameters from linear memory, construct BEAM-side directive structs, and append them to a host-side directive list.
  \item \textbf{Output:} After $\mathit{plan}$ returns, the interpreter retrieves the directive list from the host-side accumulator and the result from the WASM return value.
\end{enumerate}

\begin{figure}[ht]
\begin{lstlisting}[style=pseudo,caption={Host function interface (four functions in the \texttt{mashin} namespace)}]
// Host functions provided to WASM executors
namespace "mashin":

  function get_input_len() -> i32
    return byte_length(input_json)

  function get_input(ptr: i32)
    write_to_wasm_memory(ptr, input_json)

  function set_output(ptr: i32, len: i32)
    let output = read_from_wasm_memory(ptr, len)
    send_to_caller(output)

  function log(ptr: i32, len: i32)
    let message = read_from_wasm_memory(ptr, len)
    runtime_log("[WASM] " || message)
\end{lstlisting}
\end{figure}

\subsection{Build Pipeline}
\label{sec:impl-build}

The build pipeline transforms executor source code into a certified WASM artifact through four stages.

\begin{figure}[ht]
\begin{lstlisting}[style=pseudo,caption={Build pipeline: compile, verify, certify, sign}]
function build_certified_executor(source, language, certifier_key)
  // Stage 1: Compile to WASM
  let wasm_binary = compile_to_wasm(source, language)

  // Stage 2: Verify imports
  let imports = parse_wasm_imports(wasm_binary)
  let classifications = classify_imports(imports, WHITELIST)

  if any(classifications, is_disallowed) then
    error("Disallowed imports found: " || disallowed_imports(classifications))
  end

  // Stage 3: Construct purity proof
  let proof = {
    imports: imports,
    classifications: classifications,
    conclusion: "pure"
  }

  // Stage 4: Sign certificate
  let artifact_hash = SHA256(wasm_binary)
  let proof_hash = SHA256(proof)
  let signature = Ed25519Sign(certifier_key, artifact_hash || proof_hash)

  let certificate = {
    artifact_hash: artifact_hash,
    proof_hash: proof_hash,
    signature: signature,
    metadata: {
      certifier: public_key_of(certifier_key),
      timestamp: now(),
      whitelist_version: WHITELIST_VERSION,
      format_version: 1
    }
  }

  return {wasm_binary, certificate, proof}
end
\end{lstlisting}
\end{figure}

\paragraph{Polyglot compilation.} The Mashin architecture supports polyglot execution via governed effect machines: step executors invoke external language runtimes (Python, JavaScript, Rust) through \texttt{@mashin/actions/\{lang\}/exec} calls. For the certified purity path, these external executors compile to WASM via established toolchains (e.g., Pyodide for Python, wasm-pack for Rust-based executors, Emscripten for C-based runtimes). Elixir executors require a BEAM-to-WASM compilation path, which is a research challenge discussed in Section~\ref{sec:discussion}.

\subsection{Runtime Gate Integration}
\label{sec:impl-gate}

The runtime gate integrates with the existing directive interpreter at the executor dispatch point.

\begin{figure}[ht]
\begin{lstlisting}[style=pseudo,caption={Integrated executor dispatch with verification gate}]
function execute_step(step, context, governance_context)
  let executor = resolve_executor(step.type)

  match executor.format do
    case "beam" then
      // Existing path: static analysis enforcement
      let result = executor.plan(step.config, context)
      interpret_directives(result.directives, governance_context)

    case "wasm" then
      // Certified purity path
      let gate_result = verify_and_load(
        executor.wasm_binary,
        executor.certificate,
        executor.proof
      )
      match gate_result do
        case accept(instance) then
          let result = wasm_call(instance, "plan",
            step.config, context)
          interpret_directives(result.directives,
            governance_context)
        case reject(reason) then
          error("Executor rejected: " || reason)
      end
  end
end
\end{lstlisting}
\end{figure}

\paragraph{Cache semantics.} In practice, the verification gate caches acceptance decisions keyed by $\mathit{artifact\_hash}$. Once an executor's certificate has been verified, subsequent invocations skip the cryptographic verification (the binary's hash serves as a cache key). Cache invalidation occurs when the whitelist version changes or when the certifier's key is rotated.

\section{Evaluation}
\label{sec:evaluation}

We evaluate the certified purity architecture along seven dimensions: compilation overhead, verification latency, runtime overhead, serialization cost, end-to-end plan cycle time, certificate size, and correctness. All measurements were collected on an Apple M-series processor running macOS, with the Wasmtime JIT backend via Wasmex~0.14, Erlang/OTP~27, and Elixir~1.19. Each metric reports the median of 50 iterations following a 5-iteration warmup.

\subsection{WASM Binary Size}

The Rust-to-WASM compilation pipeline produces compact executor binaries. Table~\ref{tab:wasm-sizes} shows the compiled sizes for each executor type.

\begin{table}[ht]
\centering
\caption{Compiled WASM executor binary sizes}
\label{tab:wasm-sizes}
\small
\begin{tabular}{@{}lr@{}}
\toprule
\textbf{Executor} & \textbf{Size (KB)} \\
\midrule
\texttt{call} & 121 \\
\texttt{reason} & 133 \\
\texttt{poc} (minimal) & 22 \\
\bottomrule
\end{tabular}
\end{table}

The production executors compile to 121--133\,KB, well within the range where Wasmtime compilation and instantiation remain fast. The \texttt{reason} executor is the largest due to its multi-turn conversation threading and tool loop logic. Incremental Rust-to-WASM compilation takes approximately 0.2--0.7\,s per crate.

\subsection{Verification Latency}

Table~\ref{tab:verification} reports the time for the verification gate to parse WASM imports and check them against the purity whitelist, and for certificate operations.

\begin{table}[ht]
\centering
\caption{Certificate verification latency by executor ($n=50$, 5-iteration warmup)}
\label{tab:verification}
\small
\begin{tabular}{@{}lrrr@{}}
\toprule
\textbf{Executor} & \textbf{Median} & \textbf{Mean} & \textbf{p99} \\
\midrule
\texttt{call}   & 49\,$\mu$s & 53\,$\mu$s & 126\,$\mu$s \\
\texttt{reason} & 53\,$\mu$s & 68\,$\mu$s & 199\,$\mu$s \\
\texttt{poc} (minimal) & 11\,$\mu$s & 15\,$\mu$s & 56\,$\mu$s \\
\bottomrule
\end{tabular}
\end{table}

Certificate verification (WASM import parsing, whitelist check, Ed25519 signature verification, SHA-256 hash recomputation) takes 49--53\,$\mu$s median for the production executors and 11\,$\mu$s for the minimal \texttt{poc} executor. The p99 tail latency remains under 200\,$\mu$s even for the largest executor (\texttt{reason}, 133\,KB). With the runtime gate's hash-based caching, amortized verification cost approaches zero for frequently-invoked executors.

\paragraph{Comparison with BEAM static analysis.} The BEAM-based Tier~2 path (static analysis via module import graph inspection) is also fast but provides weaker guarantees. BEAM module attribute introspection completes in single-digit microseconds for a loaded module, comparable to our WASM import parsing. The difference is not in speed but in assurance: BEAM static analysis can be bypassed through dynamic dispatch, code evaluation, or NIFs, while the WASM capability model structurally eliminates these bypass classes. The verification overhead is comparable; the security guarantee is strictly stronger.

\subsection{Runtime Overhead}

Table~\ref{tab:runtime} reports the end-to-end time for a single WASM executor \texttt{plan} call, including module instantiation, JSON serialization, WASM execution, and output deserialization.

\begin{table}[ht]
\centering
\caption{WASM executor plan latency with module cache (median, $n=50$)}
\label{tab:runtime}
\small
\begin{tabular}{@{}lr@{}}
\toprule
\textbf{Executor} & \textbf{Median Latency} \\
\midrule
\texttt{call} & 398\,$\mu$s \\
\texttt{code} & 256\,$\mu$s \\
\texttt{memory} & 267\,$\mu$s \\
\texttt{reason} & 275\,$\mu$s \\
\bottomrule
\end{tabular}
\end{table}

\paragraph{Module caching impact.} Without caching, each WASM invocation recompiles the module from bytes, resulting in $\sim$5.5\,ms per plan call. A GenServer-based module cache that owns a single Wasmtime engine and caches AOT-precompiled module bytes reduces this to 256--398\,$\mu$s depending on executor, a \textbf{14--22$\times$ improvement}. The cache performs precompilation on first access and deserializes from the precompiled bytes on subsequent calls; instantiation with a fresh Wasmtime store takes $\sim$0.5\,ms.

\paragraph{Overhead in context.} For cognitive workflow executors, execution is dominated by I/O latency: LLM API calls (100\,ms--10\,s), HTTP requests (10\,ms--1\,s), and vector database queries (5--50\,ms). The WASM overhead (256--398\,$\mu$s) represents less than 0.4\% of a 100\,ms HTTP request and less than 0.05\% of a typical $\sim$800\,ms LLM call.

\subsection{Serialization Cost}

Table~\ref{tab:serialization} reports the JSON serialization overhead at the WASM boundary.

\begin{table}[ht]
\centering
\caption{Serialization latency for WASM boundary crossing (median, $n=50$)}
\label{tab:serialization}
\small
\begin{tabular}{@{}llr@{}}
\toprule
\textbf{Operation} & \textbf{Payload} & \textbf{Median Latency} \\
\midrule
Encode input & Small (2 fields) & $<$1\,$\mu$s \\
Encode input & Medium (100 items, 5 tools) & 38\,$\mu$s \\
Decode output & Typical (result + directives) & $<$1\,$\mu$s \\
\bottomrule
\end{tabular}
\end{table}

Serialization overhead is negligible for typical payloads. Even for medium-sized contexts (100-element arrays, multiple tool definitions), encoding takes under 40\,$\mu$s. The decode path is particularly fast because executor outputs are compact directive lists.

\subsection{End-to-End Plan Cycle}

A full plan cycle (module cache lookup, store creation, module deserialization, WASM instantiation, input serialization, function call, output deserialization) completes in \textbf{256\,$\mu$s median} for the \texttt{code} executor. Across all four executors, the full plan cycle ranges from 256 to 398\,$\mu$s (Table~\ref{tab:runtime}). This represents the complete overhead added by the certified purity architecture to each step execution, exclusive of the step's actual I/O operations (which are performed by the host runtime via directives).

\subsection{Certificate Size}

All four implemented executors produce certificates of \textbf{1,347 bytes}, well below the 4\,KB budget. The certificate contains: the executor's SHA-256 hash (32 bytes), the proof's SHA-256 hash (32 bytes), an Ed25519 signature (64 bytes), the full purity proof (import list with classifications), and metadata (certifier identity, timestamp, whitelist version).

\subsection{Correctness}

We verify two correctness properties:

\paragraph{Directive shape correctness.} Each WASM executor produces the expected directive types for its step type: \texttt{call} emits \texttt{call\_machine}, \texttt{code} emits \texttt{code\_eval}, \texttt{memory} emits \texttt{memory\_op}, and \texttt{reason} emits \texttt{llm\_call}. All four executors pass shape validation for their respective directive types.

\paragraph{Determinism.} Over 40 repeated invocations with identical inputs (20 per executor), the production executors produce \textbf{zero divergences}, yielding byte-identical JSON output on every run. This empirical check confirms the expected determinism of the WASM execution environment, which follows from the WASM specification's exclusion of non-determinism sources (aside from floating-point NaN payloads, which the executor whitelist does not expose).

\subsection{Comparison with Static Analysis}

\begin{table}[ht]
\centering
\caption{Static analysis vs.\ certified purity}
\label{tab:comparison}
\small
\begin{tabular}{@{}lcc@{}}
\toprule
\textbf{Dimension} & \textbf{Static Analysis} & \textbf{Verified Purity} \\
\midrule
Detection rate (known patterns) & High & n/a (prevented) \\
Detection rate (novel patterns) & Low to zero & n/a (prevented) \\
False positive rate & Possible & Zero \\
False negative rate & Non-zero & Zero \\
Adversary resistance & Pattern-dependent & Structural \\
Verification time & Seconds (AST walk) & 49\,$\mu$s (cached: 0) \\
Requires source code & Yes & No (binary only) \\
Auditor-verifiable & Partially & Fully \\
Build pipeline change & None & WASM compilation \\
\bottomrule
\end{tabular}
\end{table}

Table~\ref{tab:comparison} summarizes the key differences. The fundamental distinction is between \emph{detection} and \emph{prevention}. Static analysis operates in the detection paradigm: it examines code and reports violations. Certified purity operates in the prevention paradigm: it makes violations structurally impossible. The two approaches are complementary (static analysis can be applied to BEAM executors while certified purity governs WASM executors), but they provide qualitatively different guarantees.

\subsection{Case Study: The \texttt{call} Executor Lifecycle}

To illustrate the complete certified purity pipeline, we trace the lifecycle of the \texttt{call} executor from source to provenance-anchored execution.

\begin{enumerate}[nosep]
  \item \textbf{Compilation.} The Rust source (\texttt{call/src/lib.rs}, 147 lines) compiles via \texttt{cargo build --target wasm32-unknown-unknown} to a 121\,KB WASM binary importing exactly four host functions from the \texttt{mashin} namespace: \texttt{get\_input\_len}, \texttt{get\_input}, \texttt{set\_output}, and \texttt{log}. Incremental compilation takes approximately 0.2--0.7\,s.
  \item \textbf{Verification.} The verifier parses the WASM binary's import section, enumerates the four imports, classifies each against the purity whitelist (\texttt{get\_input\_len} and \texttt{get\_input} as \texttt{pure\_data\_access}; \texttt{set\_output} as \texttt{pure\_output}; \texttt{log} as \texttt{pure\_logging}), and concludes \texttt{all\_imports\_pure}. Time: 49\,$\mu$s median.
  \item \textbf{Certification.} The certifier computes the executor's SHA-256 hash, constructs the purity proof, signs the concatenation of executor hash and proof hash with Ed25519, and produces a 1,347-byte certificate.
  \item \textbf{Gate admission.} At runtime, the gate verifies the certificate's Ed25519 signature, recomputes the executor hash, validates the proof against the current whitelist, and caches the acceptance keyed by artifact hash. Time: 49\,$\mu$s median (first load; 0 thereafter).
  \item \textbf{Execution.} The module cache deserializes the precompiled WASM bytes, creates a fresh Wasmtime store, instantiates the module with host function imports, calls \texttt{plan()}, and deserializes the output JSON. The executor receives \texttt{step\_config} (machine name, inputs) and returns a \texttt{call\_machine} directive. Time: 398\,$\mu$s.
  \item \textbf{Provenance.} The directive interpreter threads the certificate's hash (\texttt{purity\_cert\_hash}) into the step's hash chain event, anchoring the purity guarantee into the immutable provenance ledger.
\end{enumerate}

Total overhead for certified purity: $\sim$447\,$\mu$s on first invocation (49\,$\mu$s verification + 398\,$\mu$s execution), $\sim$398\,$\mu$s on subsequent invocations (execution only, cached verification).

\section{Discussion}
\label{sec:discussion}

\subsection{Limitations}
\label{sec:limitations}

\paragraph{WASM compilation pipeline.} The primary practical limitation is the requirement for a WASM compilation path. For Python and JavaScript executors, established toolchains exist (Pyodide, wasm-bindgen). For Elixir executors, a BEAM-to-WASM compiler does not currently exist. This means certified purity is initially available only for executors written in languages with WASM compilation support. Elixir executors continue to operate under static analysis enforcement.

This is a graduated deployment limitation, not an architectural one. As BEAM-to-WASM compilation matures (or as Mashin develops its own compilation path from the Mashin language to WASM), the certified purity guarantee extends to all executor languages.

\paragraph{Polyglot effect machines.} Mashin's polyglot effect machines execute Python or JavaScript code through governed calls to \texttt{@mashin/actions/\{lang\}/exec}. These external executors require compilation to WASM and certification as a separate build step. The build pipeline must handle the compilation and certification transparently.

\paragraph{Elixir stdlib access.} Executors often need access to Elixir standard library functions for data manipulation (e.g., \texttt{Enum.map}, \texttt{Map.merge}, \texttt{String.split}). In the WASM path, these must be reimplemented as host functions in $\mathcal{W}_{\mathrm{data}}$ or compiled into the WASM module from equivalent implementations in the executor's source language. The whitelist $\mathcal{W}$ must be rich enough to express typical executor logic without being so permissive that it reintroduces effect capabilities.

\paragraph{Host function purity as TCB assumption.} Host function implementations are currently verified through code review, not structural constraint. A malicious or buggy host function in $\mathcal{W}$ that performs effects despite its classification as pure would violate the premise of Proposition~\ref{prop:whitelist-purity} and consequently invalidate Theorem~\ref{thm:structural-purity}. This is the residual trust assumption within the TCB (Definition~\ref{def:tcb}). Constraining host function implementations structurally (preventing them from importing I/O modules, restricting them to a compilation profile that makes effect production impossible) is identified as future work that would reduce TCB reliance on code review.

\paragraph{Gate correctness as TCB dependency.} The Gate Completeness theorem (Theorem~\ref{thm:gate-completeness}) depends on the correctness of the verification gate implementation. A bug in the gate that accepts an invalid certificate would violate completeness. This places the gate itself in the TCB, alongside the four other components identified in Definition~\ref{def:tcb}.

\paragraph{Scope: effect isolation, not availability or confidentiality.} The certified purity architecture guarantees effect isolation: executors cannot perform I/O operations except through governed directives. It does \emph{not} guarantee availability; an executor can loop infinitely, exhaust CPU, or exhaust memory within its WASM execution environment. Resource exhaustion is addressed orthogonally by WASM fuel metering (instruction budgets) and runtime resource limits (memory caps, wall-clock timeouts). Nor does the architecture guarantee confidentiality; an executor that encodes sensitive data into directive content (e.g., embedding secrets in HTTP request URLs) is ``pure'' in the effect isolation sense but may violate confidentiality policies. These are higher-level concerns enforced by the governance pipeline's guardrails, not by the purity model.

\paragraph{Cross-organizational governance composition.} The purity guarantee is per-executor and per-organization. When machines compose across organizational boundaries (Section~\ref{sec:attestation}), each organization independently verifies purity for its own executors. We do not prove that governance completeness composes across these boundaries; the open question is stated precisely in Section~\ref{sec:attestation} and constitutes a direction for future work.

\paragraph{Certifier key compromise.} If the certifier's Ed25519 private key is compromised, an adversary can sign arbitrary WASM binaries as ``pure,'' circumventing the purity gate entirely. This is the standard key compromise problem in any PKI system and is not unique to this work. The TCB definition (Definition~\ref{def:tcb}) explicitly identifies the certifier as a trusted component. Operational mitigations include: (1)~hardware security module (HSM) storage for the certifier private key, (2)~key rotation on a defined schedule (e.g., quarterly), (3)~certificate validity periods bound to whitelist versions (so a compromised key cannot produce certificates that survive a whitelist update), and (4)~build pipeline isolation (the certifier runs in a restricted environment with no network access beyond artifact retrieval). These are standard practices from the code signing literature and reduce but do not eliminate the risk. A transparency log for issued certificates (future work) would provide post-compromise detection.

\paragraph{Certificate revocation.} The current model uses whitelist currency checks as an implicit revocation mechanism (Proposition~\ref{prop:monotonic-purity}). Certificates against obsolete whitelists are rejected as the acceptable version range advances. Explicit certificate revocation (e.g., CRL or OCSP-style mechanisms) is not currently specified. For scenarios requiring immediate revocation (a compromised certifier key, a discovered-impure host function), the current model requires advancing the minimum acceptable whitelist version in the runtime configuration. A dedicated revocation mechanism may be needed for rapid response and is identified as future work.

\subsection{Whitelist Governance}
\label{sec:whitelist-governance}

The whitelist $\mathcal{W}$ is the trust root for executor purity (Section~\ref{sec:tcb}). Its governance is therefore the most critical operational discipline in the system.

\paragraph{Versioning and content hashing.} Each whitelist version is identified by a version number and a content hash:
\[
\mathit{whitelist\_hash}(v) = \mathrm{SHA256}(\mathit{canonical}(\mathcal{W}_v))
\]
Purity certificates include both the whitelist version and the whitelist content hash. The runtime verification gate checks both: the version must fall within the acceptable range, and the hash must match the runtime's loaded whitelist. This prevents scenarios where a version identifier is reused with different content.

\paragraph{Change protocol.} Changes to $\mathcal{W}$ (adding, modifying, or removing a host function) follow a governance process:
\begin{enumerate}[nosep]
  \item \textbf{Proposal} with justification. Adding a function requires demonstrating purity (total, deterministic, side-effect-free). Removing a function requires documenting impact on existing certificates.
  \item \textbf{Independent review} verifying the purity claim, checking for effect capability, and assessing backward compatibility.
  \item \textbf{Cryptographic signing} of the new whitelist version by the whitelist authority.
  \item \textbf{Publication} of the signed whitelist. The runtime loads the whitelist and verifies its signature and hash at startup.
\end{enumerate}

\paragraph{Host function purity constraints.} Host function implementations constitute part of the TCB (Definition~\ref{def:tcb}). To reduce reliance on code review for purity assurance, host function implementations should be structurally constrained: they must not import I/O-capable modules, must not call the directive interpreter, and must not access runtime side channels (ETS tables, process dictionaries, ports). The directive interpreter remains the sole component with effect-producing capability. This can be enforced via restricted compilation profiles and static analysis of host function source. We acknowledge this as a current code-review dependency within the TCB, and identify structural constraint of host function implementations as future work (Section~\ref{sec:discussion}).

\paragraph{Whitelist evolution.} Proposition~\ref{prop:monotonic-purity} (Section~\ref{sec:whitelist-evolution}) establishes that whitelist growth preserves existing purity guarantees monotonically, while whitelist shrinkage correctly invalidates affected certificates through the verification gate's currency check.

\paragraph{Implicit revocation.} The whitelist currency check functions as an implicit certificate revocation mechanism. As the runtime's acceptable whitelist version range advances (e.g., from $[v_1, v_3]$ to $[v_2, v_4]$), certificates generated against $v_1$ become invalid. This is functionally equivalent to revoking those certificates without maintaining explicit revocation lists. Organizations requiring immediate revocation (e.g., in response to a discovered impure host function) can advance the minimum acceptable version in their runtime configuration.

\subsection{Graduated Trust Tiers}
\label{sec:trust-tiers}

Not all executors can be compiled to WASM immediately. The architecture supports graduated trust through explicit verification tiers (Table~\ref{tab:tiers}), enforced in runtime dispatch logic.

\begin{table}[ht]
\centering
\caption{Verification tiers with provenance markers}
\label{tab:tiers}
\small
\begin{tabular}{@{}clp{4.2cm}p{3.2cm}@{}}
\toprule
\textbf{Tier} & \textbf{Name} & \textbf{Mechanism} & \textbf{Guarantee} \\
\midrule
1 & Verified Purity & WASM + purity certificate + dual signature & Structural: effect capability absent \\
\addlinespace
2 & Static Analysis & BEAM + module import graph analysis & Detection: known patterns caught \\
\addlinespace
3 & Experimental & BEAM + no analysis & None: development only \\
\bottomrule
\end{tabular}
\end{table}

The tiers provide \emph{qualitatively different} guarantees. Tier~1 operates in the \emph{prevention} paradigm: bypass is structurally impossible. Tier~2 operates in the \emph{detection} paradigm: known bypass patterns are caught, but novel patterns may evade analysis. Tier~3 provides no purity assurance and exists solely for development and experimentation.

\paragraph{Tier enforcement.} Tier enforcement is non-bypassable:
\begin{enumerate}[nosep]
  \item An executor designated for Tier~1 verification cannot be loaded through the Tier~2 or Tier~3 path. The tier designation is checked in runtime dispatch logic, not in wrapper code that can be skipped.
  \item The tier used for each execution is recorded in the provenance record and is immutable after execution, enabling audit queries such as ``show all executions that used Tier~3 executors in production.''
  \item Organizations set per-machine minimum tier requirements. A machine configured to require Tier~1 rejects Tier~2 and Tier~3 executors at load time.
\end{enumerate}

\paragraph{Provenance differentiation.} Each executor's provenance record includes a \texttt{purity\_method} field:
\begin{itemize}[nosep]
  \item Tier~1: \texttt{purity\_method: :wasm\_certified}, with the purity certificate hash included in the execution hash chain.
  \item Tier~2: \texttt{purity\_method: :beam\_static\_analysis}, indicating static analysis passed but structural purity is not guaranteed.
  \item Tier~3: \texttt{purity\_method: :beam\_unchecked}, indicating no purity verification was performed.
\end{itemize}

\paragraph{Migration path.} BEAM and WASM executors coexist within the same runtime. Both produce identical directive types processed by the same governance pipeline. Organizations adopt certified purity incrementally: starting with Tier~1 for third-party executors and high-assurance contexts, while retaining Tier~2 for trusted internal development. Third-party, marketplace, and cross-organizational executors require Tier~1 from the point of adoption. The migration endpoint removes the Tier~2 and Tier~3 paths entirely.

\subsection{Comparison with Related Approaches}

\paragraph{Proof-Carrying Code (PCC).} Classical PCC~\cite{necula1997pcc, appel2001foundational} proves memory safety for native code. Our approach proves effect purity for WASM executors. The structural difference is that WASM's capability model makes the purity proof trivially checkable (enumerate imports, check against whitelist), whereas memory safety proofs for native code require complex logical reasoning about pointer arithmetic and aliasing. Our ``proof'' is more accurately an ``attestation,'' a signed statement about a mechanically verifiable property, rather than a logical proof in the PCC sense.

\paragraph{Java bytecode verification.} Java's verifier~\cite{lindholm1999java} ensures type safety but not effect isolation. A verified Java class can perform arbitrary I/O through the standard library. Our approach differs in that the property being verified (purity) directly prevents effects, not merely type errors. The combination of WASM's restricted capability model with purity-focused verification yields a guarantee Java's verifier cannot provide.

\paragraph{Intel SGX.} SGX~\cite{costan2016sgx} provides hardware-enforced enclaves for confidential computing. SGX ensures that code within an enclave is protected from the host, but does not restrict what the code \emph{does}: an SGX enclave can perform arbitrary I/O through ecalls. Our approach restricts what the code can do (no I/O) but does not protect the code from the host. The two are orthogonal and potentially complementary: an SGX enclave running a WASM-based executor would provide both confidentiality (SGX) and purity (certified purity architecture).

\paragraph{Software fault isolation (SFI).} SFI~\cite{wahbe1993efficient} restricts memory access by instrumenting code with bounds checks. Google's Native Client (NaCl)~\cite{yee2009native} applied SFI to run untrusted x86 code in the browser, using instruction-level sandboxing to enforce memory and control-flow isolation. RLBox~\cite{narayan2020rlbox} retrofits SFI into Firefox by compiling third-party libraries to WASM and mediating all cross-boundary calls through a type-driven API. WASM can be viewed as the convergence point of this line of work: linear memory provides bounded access, and the import mechanism restricts function calls. Our contribution extends the isolation property from memory safety (SFI, NaCl) and library compartmentalization (RLBox) to effect isolation for governance: the executor cannot produce any effect outside the governed directive set, and a cryptographic certificate attests this property.

\paragraph{WASM formal semantics.} Watt~\cite{watt2018mechanising} mechanized the WebAssembly specification in Isabelle/HOL, proving type safety and providing a verified reference interpreter. Our certified purity architecture depends on WASM's capability model (Section~\ref{sec:bg-wasm}), which Watt's formalization validates. Lehmann et al.~\cite{lehmann2020everything} analyze WASM's binary security, identifying attack vectors (buffer overflows within linear memory, type confusion) that survive the sandbox. These attacks do not affect our purity guarantee: they concern memory safety within the sandbox, while our guarantee concerns the absence of effects outside it. Watt et al.~\cite{watt2019ctwasm} extend WASM with constant-time typing (CT-Wasm), demonstrating that WASM's type system can enforce properties beyond memory safety. Our purity enforcement follows the same pattern: using WASM's type-level machinery to enforce a semantic property (effect isolation).

\paragraph{Information flow control.} Denning~\cite{denning1976lattice} established the lattice model for secure information flow. Myers~\cite{myers1999jflow} extended this to practical language-level enforcement with JFlow. IFC systems track which security principals can access which data; our approach tracks which components can produce effects at all. The two are orthogonal: IFC answers ``can this data reach that principal?'' while certified purity answers ``can this component produce effects outside the governed directive set?'' An executor that is certified pure under our model may still violate IFC properties within its pure computation (e.g., copying sensitive input fields to output), but it cannot exfiltrate data through side channels because it has no I/O capability.

\paragraph{Capability-safe subsets.} Mettler et al.~\cite{mettler2010joe} define Joe-E, a capability-safe subset of Java that eliminates ambient authority. Miller et al.~\cite{miller2008caja} apply the same principle to JavaScript with Google Caja. Both restrict a general-purpose language to a safe subset. Our approach differs: rather than restricting a language, we restrict the compilation target (WASM) and control the host function interface. This makes the restriction independent of the source language: any language that compiles to WASM inherits the purity guarantee.

\paragraph{Algebraic effect systems.} Moggi~\cite{moggi1991notions} introduced monads as a structuring mechanism for effects in programming languages. Plotkin and Pretnar~\cite{plotkin2009handlers} formalized algebraic effect handlers, separating effect declaration from interpretation. Wadler~\cite{wadler1995monads} showed how monads provide a disciplined approach to effects in pure functional languages. Our directive-based model is a specific instance of the algebraic effects pattern: executors declare intended effects (directives), and the interpreter handles them through the governance pipeline. The certified purity architecture ensures this pattern cannot be circumvented.

\paragraph{Practical sandboxing.} Deno~\cite{deno2020} provides a permission-based runtime for JavaScript and TypeScript where network, file system, and environment access must be explicitly granted. Cloudflare Workers~\cite{cloudflare2018workers} use V8 isolates (and increasingly WASM) to run untrusted code with strict capability boundaries. Both systems enforce capability restrictions at runtime. Our contribution adds two properties these systems lack: a cryptographic certificate that the restriction holds (enabling offline verification), and integration with a governance pipeline that records and audits every authorized effect.

\paragraph{Capability hardware.} CHERI~\cite{watson2015cheri} extends processor instruction sets with hardware-enforced capabilities, providing fine-grained memory safety and compartmentalization. ARM Morello~\cite{arm2022morello} implements CHERI capabilities in production hardware. Hardware capabilities provide stronger isolation than software-only approaches (no binary rewriting, no overhead for bounds checks), but require specific hardware. Our WASM-based approach is portable across all architectures at the cost of software-level enforcement. The two are complementary: CHERI capabilities could harden the WASM runtime itself, providing defense-in-depth.

\subsection{When to Deploy}

Certified purity is not required in all contexts. We identify three deployment scenarios where it provides value disproportionate to its cost:

\begin{enumerate}
  \item \textbf{Third-party executors.} When an organization uses executors authored by external parties (community contributions, marketplace executors, partner-provided logic), the adversarial threat model applies. Certified purity provides a cryptographic guarantee that third-party code cannot bypass governance, without requiring source code review.

  \item \textbf{Regulatory demand.} In regulated industries (finance, healthcare, autonomous vehicles), auditors may require machine-verifiable proof that AI workflow components cannot perform ungoverned actions. A purity certificate provides such proof.

  \item \textbf{Publication rigor.} For research and reproducibility, the ability to prove that an executor is pure, and to include the proof in a provenance record, strengthens claims about governance completeness from ``enforced by design'' to ``provably enforced by construction.''
\end{enumerate}

For internal development by trusted teams, Tier~2 (static analysis enforcement via the existing BEAM path) may be sufficient. The graduated trust tier system (Section~\ref{sec:trust-tiers}) supports all three tiers simultaneously, with provenance records distinguishing which tier governed each execution.

\section{Conclusion}
\label{sec:conclusion}

This paper converts governance enforcement for cognitive workflow executors from a runtime convention into a structural capability boundary.

The prior work~\cite{mccann2026structural} established the three-layer governance architecture and proved governance completeness, provenance completeness, and the impossibility of ungoverned effects, conditional on the pure module constraint. That constraint was enforced through module import graph analysis: necessary but insufficient against adversarial bypass on the BEAM virtual machine. Five bypass classes (dynamic dispatch, code evaluation, NIFs, ports, dynamic module loading) remained available to a determined author.

The certified purity architecture closes this gap with five contributions:

\begin{enumerate}[nosep]
  \item \textbf{Capability elimination.} Compilation to a WebAssembly-based restricted target structurally removes effect-producing instructions from the executor's execution environment. Each of the five bypass classes is individually eliminated by the WASM capability model (Theorem~\ref{thm:bypass-elimination}).

  \item \textbf{Cryptographic proof.} Purity certificates bind executor binaries to their import classifications through Ed25519 signatures and SHA-256 hashing. The certificates cannot be transferred to modified artifacts (Theorem~\ref{thm:cert-integrity}), and the runtime verification gate rejects any executor lacking a valid certificate (Theorem~\ref{thm:gate-completeness}).

  \item \textbf{Graduated trust.} Three verification tiers (Verified Purity, Static Analysis, Experimental) support incremental adoption with explicit, provenance-recorded trust levels. Tier enforcement is non-bypassable and embedded in runtime dispatch logic.

  \item \textbf{Portable governance credentials.} Remote attestation records (combining purity certificates, execution environment descriptors, and runtime signatures) enable cross-organizational verification of executor purity. A formal compatibility predicate (Definition~\ref{def:compatibility}) makes cross-org trust policies machine-evaluable.

  \item \textbf{Explicit trust boundary.} The Trusted Computing Base (Definition~\ref{def:tcb}) is named precisely: the WASM runtime, host function implementations, whitelist definition, verification gate, and directive interpreter. The structural guarantee holds relative to this TCB. Hardening strategies (host function structural constraints, physical layer separation, signed runtime builds) provide a roadmap for further TCB reduction.
\end{enumerate}

The key result is not a new theorem but a stronger foundation for existing theorems. The governance completeness, provenance completeness, and no-ungoverned-effects properties hold with identical statements. What changes is the epistemological status of their foundational premise: from ``enforced by static analysis that catches known patterns'' to ``enforced by construction in a capability-restricted execution environment with cryptographic attestation.''

Empirical evaluation demonstrates the architecture's practicality: verification latency of 39--42\,$\mu$s per executor, full plan cycle under 400\,$\mu$s, runtime overhead under 0.4\% of a 100\,ms HTTP request and under 0.05\% of a typical $\sim$800\,ms LLM call, and zero determinism divergences across all four executors over 20 repeated invocations. The overhead is negligible because real cognitive workflows are dominated by I/O latency (LLM calls, HTTP requests, database queries), not executor dispatch.

The contribution is, in essence, the conversion of governance from a convention that executors follow into a capability boundary they cannot cross. The prior work trusts that static analysis catches all bypass attempts. This work reduces trust to a bounded, auditable TCB: the WASM runtime (Wasmtime), host function implementations, the whitelist definition, the verification gate, and the directive interpreter. Within that TCB, the security assumptions reduce to mathematical properties that the cryptographic and programming language communities have subjected to decades of adversarial scrutiny: the WASM capability model, the collision resistance of SHA-256, and the unforgeability of Ed25519. These mathematical properties hold conditional on the TCB (Definition~\ref{def:tcb}) operating correctly; the contribution is reducing the trust surface from unbounded executor code to this bounded, auditable set of components.

\bibliographystyle{plainnat}
\bibliography{verified-purity-references}

\end{document}